\documentclass[a4paper,english,cleveref, autoref, thm-restate]{lipics-v2021}

\pdfoutput=1 %
\hideLIPIcs  %

\usepackage[utf8]{inputenc}
\usepackage{mathtools}
\usepackage{xparse}
\usepackage{todonotes}
\usepackage[numbers,sort]{natbib}
\usepackage{algorithm}
\usepackage[noend]{algpseudocode}
\algnewcommand\algorithmicinput{\textbf{Input:}}
\algnewcommand\algorithmicoutput{\textbf{Output:}}
\algnewcommand\Input{\item[\algorithmicinput]}
\algnewcommand\Output{\item[\algorithmicoutput]}

\usepackage{tabularx,booktabs,multirow} 
\usepackage{array, makecell}
\usepackage[shortlabels]{enumitem}
\usepackage{amsthm}
\usepackage{amsmath}
\usepackage{amssymb}
\usepackage{amsfonts}

\usepackage{tikz}
\usetikzlibrary{arrows,decorations.pathmorphing,decorations.pathreplacing,backgrounds,positioning,fit,matrix,shapes,calc,decorations.text}

\tikzset{
	position/.style args={#1:#2 from #3}{
		at=(#3.#1), anchor=#1+180, shift=(#1:#2)
	},
	vertex/.style = {circle, minimum size=16pt, draw, inner sep=1pt},
	timearc/.style = {draw=black,->,shorten <=1.5pt, shorten >=1.5pt,>=stealth',
		nodes={sloped,font=\scriptsize},	
	},
}

\usepackage{subcaption}

\usepackage{doi} %

\newtheorem{invariant}[theorem]{Invariant}

\Crefname{observation}{\text{Observation}}{\text{Observations}}

\newcommand{\NN}{\mathbb{N}} %

\newcommand{\bigO}{\mathcal{O}} %

\newcommand\abs[1]{\left|#1\right|} %

\newcommand{\GG}{\mathcal{G}} %

\newcommand\drc{\textsc{Delay-Robust Connection}}

\newcommand{\drgame}{Delayed-Routing Game}
\newcommand{\drpgame}{Delayed-Routing Path Game}
\newcommand\drg{\textsc{Delayed-Routing Game}}
\newcommand\drpg{\textsc{Delayed-Routing Path Game}}

\newcommand\qbf{\textsc{QBF}}
\newcommand\qbfg{\textsc{QBF Game}}

\newcommand{\lepoly}{\le_{\textnormal{m}}^{\textnormal{poly}}}

\newcommand{\CNP}{\textnormal{NP}}
\newcommand{\CPSPACE}{\textnormal{PSPACE}}

\newcommand{\CSHRPP}{\textnormal{\#P}}

\DeclarePairedDelimiterX{\set}[1]{\{ }{ \} }{\setargs{#1}}
\NewDocumentCommand{\setargs}{>{\SplitArgument{1}{;}}m}
{\setargsaux#1}
\NewDocumentCommand{\setargsaux}{mm}
{\IfNoValueTF{#2}{#1} {#1\,\delimsize|\,\mathopen{}#2}}%

\DeclareMathOperator*{\argmax}{arg\,max}
\DeclareMathOperator*{\argmin}{arg\,min}

\DeclareMathOperator{\tastart}{start}
\DeclareMathOperator{\taend}{end}
\DeclareMathOperator{\tatime}{t}
\DeclareMathOperator{\tadest}{\taend}
\DeclareMathOperator{\tatrav}{\lambda}

\newcommand{\true}{\texttt{true}}
\newcommand{\false}{\texttt{false}}

\newcommand{\mvert}{\;\middle\vert\;}

\newcommand{\problemdef}[3]{
		\begin{center}
	\begin{minipage}{0.95\textwidth}
		\noindent
		\textsc{#1}
				\vspace{5pt}\\
				\setlength{\tabcolsep}{3pt}
				\begin{tabularx}{\textwidth}{@{}lX@{}}
						\textbf{Input:} 		& #2 \\
						\textbf{Question:} 	& #3
					\end{tabularx}
	\end{minipage}
		\end{center}
}

\title{Temporal Connectivity: Coping with Foreseen and Unforeseen Delays}

\author{Eugen~F\"uchsle}{TU Berlin, Faculty IV, Algorithmics and Computational Complexity, Germany}{fuechsle@campus.tu-berlin.de }{}{}

\author{Hendrik~Molter}{Department of Industrial Engineering and Management, Ben-Gurion~University~of~the~Negev, 
Beer-Sheva, 
Israel}{molterh@post.bgu.ac.il}{https://orcid.org/0000-0002-4590-798X}{Supported by the ISF, grant No.~1070/20.}

\author{Rolf~Niedermeier}{TU Berlin, Faculty IV, Algorithmics and Computational Complexity, Germany}{rolf.niedermeier@tu-berlin.de}{https://orcid.org/0000-0003-1703-1236}{}

\author{Malte~Renken}{TU Berlin, Faculty IV, Algorithmics and Computational Complexity, Germany}{m.renken@tu-berlin.de}{https://orcid.org/0000-0002-1450-1901}{Supported by the DFG, project MATE (NI 369/17).}

\authorrunning{Eugen F\"uchsle, Hendrik Molter, Rolf Niedermeier, and Malte Renken} %

\Copyright{Eugen F\"uchsle, Hendrik Molter, Rolf Niedermeier, and Malte Renken} %



\category{} %

\relatedversion{} %

\nolinenumbers %

\begin{document}

\maketitle              %

\begin{abstract}
Consider planning a trip in a train network.
In contrast to, say, a road network, the edges are \emph{temporal},
i.e., they are only available at certain times.
Another important difficulty is that trains, unfortunately, sometimes get delayed.
This is especially bad if it causes one to miss subsequent trains.
The best way to prepare against this is to have a connection that is \emph{robust}
to some number of (small) delays.
An important factor in determining the robustness of a connection is how far in advance delays are announced.
	We give polynomial-time algorithms for the two extreme cases: delays known before departure and delays occurring without prior warning (the latter 
leading to a two-player game scenario).
Interestingly, in the latter case, we show that the problem becomes PSPACE-complete if the itinerary is demanded to be a path.
\end{abstract}

\section{Introduction}\label{sec:intro}
Computing \emph{temporal paths} is one of the back-bone algorithmic problems in 
the context of temporal graphs, that is, graphs
whose edges are present only at certain, known points in time~\cite{XuanFJ03,WuCKHHW16EffTempPath}.
Temporal graphs are specified by a set~$V$ of vertices and a set~$E$ of time arcs,
where each time arc~$(v,w,t,\lambda) \in E$ consists of a \emph{start vertex}~$v$, an \emph{end vertex}~$w$, a \emph{time label}~$t$, and a \emph{traversal time}~$\lambda$;
then there is a (direct) connection from~$v$ to~$w$ starting at time~$t$ and arriving at time~$t+\lambda$.
Temporal graphs model numerous 
real-world scenarios~\cite{HotTopicHol15,HotTopicHS19,HotTopicCas+12,HotTopicLVM18}:
Social, communication, transportation, and many other networks are usually not static 
but vary over time.

The added dimension of time causes many aspects of connectivity to behave quite differently from static (i.e., non-temporal) graphs.
Thus, the flow of items through a temporal network has to be time-respecting. 
More specifically, it follows a \emph{temporal walk} (or \emph{path}, if every vertex is visited at most once),
i.e., a sequence of time arcs $(v_i, w_i, t_i, \lambda_i)_{i=1}^{\ell}$
where $v_{i+1} = w_i$ and $t_{i+1} \geq t_i + \lambda_i$ for all $i < \ell$.
While inheriting many properties from static walks, temporal walks exhibit certain characteristics that add a further level of difficulty to algorithmic problems centered around them.
For example, temporal connectivity is not transitive: the existence of a temporal walk from vertex $u$ to $v$ and a temporal walk from $v$ to $w$ does not imply the existence of a temporal walk from $u$ to $w$.
Moreover, the temporal setting naturally leads to several notions of what an ``optimal'' temporal path could be~\cite{BentertHNN20TempWalkWaitingTime,XuanFJ03}.

Computation of temporal paths and walks has already been studied intensively~\cite{WuCKHHW16EffTempPath,XuanFJ03},
including specialized settings that are novel to temporal graphs:
For example, \citet{BentertHNN20TempWalkWaitingTime} and \citet{TemporalPathsUnderWaitingTime}
studied temporal walks and paths that are only allowed to have limited waiting time at any vertex.

In this work, we investigate another natural, inherently temporal connectivity problem. It addresses \emph{delays}.
In many real-world temporal networks such as transport networks (e.g.\ trains, shipping routes),
individual edges may get delayed for various reasons.
Thus it is an important question whether connectivity between a start and a target node is \emph{fragile},
i.e., easily disrupted by delays,
or whether it is \emph{robust}.

An important aspect in this matter is the time at which delays become known.
The earlier they are announced, the easier one can still adapt the chosen route.
In this work, we study the two endpoints of this spectrum:
one where all delays are known up front,
and one where delays occur without any prior warning.

We now briefly describe our models for these two problems,
beginning with the problem variant in which all delays are known up front.
Herein, a \emph{$D$-delayed temporal path} refers to a temporal path that remains valid when the time arcs in $D \subseteq E$ have been delayed by some fixed amount each.
A more formal definition will be given in \cref{sec:prelims}.

\problemdef{\drc{}}
{A temporal graph $\mathcal{G} = (V,E)$, two vertices $s,z \in V$, and $x,\delta \in \mathbb{N}$.}
{Is there, for every delay set~$D \subseteq E$ of size $\abs{D} \leq x$, a $D$-delayed temporal path from~$s$ to~$z$ in~$\GG$?}

In contrast, if the delays occur without prior warning, then the resulting problem is best modeled as a two-player game,
which we call the \drgame{}.
The first player (the \emph{traveler}) starts at vertex~$s$ and has to decide at each turn which time arc they want to traverse next.
The other player (the \emph{adversary}) then gets to decide whether that time arc is delayed or not.
As before, there is a bound~$x$ on the overall number of time arcs that can be delayed.
The traveler wins if they reach the target vertex~$z$.
A \emph{winning} strategy for the traveler is a strategy that guarantees that they will reach their target.

\problemdef{\drg{}}
{A temporal graph $\mathcal{G} = (V,E)$, two vertices $s,z \in V$, and $x,\delta \in \mathbb{N}$.}
{Does the traveler have a winning strategy for the \drgame{}?}

The difference between the two models \drc{} and \drg{} is illustrated in \cref{fig:example_drc_drg}.

\begin{figure}[t]
	\centering
	\tikzstyle{alter}=[circle, minimum size=16pt, draw, inner sep=1pt] 
	\tikzstyle{majarr}=[draw=black]
	
	\begin{tikzpicture}
		\begin{scope}[every node/.style=vertex]
			\node at (0,0) (a) {$s$};
			\node[position=30:8ex from a] (b) {$a$};
			\node[position=-30:8ex from a] (c) {$b$};
			\node[position=-30:8ex from b] (z) {$z$};
		\end{scope}
		
		\begin{scope}[timearc]
			\draw
				(a) edge[edge label={$(1,1)$}] (b)
				(a) edge[edge label={$(1,1)$}] (c)
				(b) edge[edge label={$(2,1)$}] (z)
				(c) edge[edge label={$(2,1)$}] (z)
				;
		\end{scope}	
	\end{tikzpicture}  	
	\caption{
		An example temporal graph, time arcs are labeled by $(\tatime(e), \tatrav(e))$.
	 	For $x=1$ and $\delta = 1$, this is a yes-instance of \drc{}:
		If the time arc from~$s$ to~$a$ is delayed then there is a temporal path from~$s$ to~$z$ via~$b$.
		Otherwise, if that time arc is not delayed, then the temporal path via $a$ is available.
		However, the same setting is a no-instance of \drg{}:
		If the traveler picks the time arc from~$s$ to~$a$, then the adversary will delay it,
		rendering the traveler stuck at~$a$ since they reach it at time~$3$.
		If the traveler picks the time arc from~$s$ to~$b$, then the situation is analogous.
	}
	\label{fig:example_drc_drg}
\end{figure}

Finally, we want to consider a variant of \drg{}, in which the traveler may not visit any vertex more than once.
We refer to this as \drpg{}. %

\subparagraph*{Related Work.} 
There has been extensive research on many other connectivity-related problems on temporal graphs~\cite{BussMNR20,Erlebach0K21,KlobasMMNZ21,staticExp3/DBLP:journals/algorithmica/MertziosMS19,EnrightMMZ19,enright2021assigning,FluschnikMNRZ20,staticExp1:journals/jcss/ZschocheFMN20,staticExp2/DBLP:journals/jcss/KempeKK02}.
Delays in temporal graphs have been considered 
in terms of manipulating
reachability sets~\cite{DBLP:conf/aaai/DeligkasP20,DBLP:journals/corr/MolTerRenkenZschoche}. 
An individual delay operation considered in the mentioned work delays a single time arc and is similar to our notion. %
Typically, the computational problems in this context are NP-hard and can be also considered as computing ``robustness measures'' for the connectivity in temporal graphs.

In companion work \cite{DelayRobustRoutes} we investigate a
problem located somewhat ``between'' \drc{} and \drg{}
in which the delays become known after the sequence of vertices to be traversed from~$s$ to~$z$ is fixed,
but before the exact time arcs to be traversed are chosen.
There, we show that this problem is NP-hard and further study its parameterized complexity.

\subparagraph*{Our Contribution.}
We introduce two computational problems related to testing connectivity between two terminal vertices in the presence of delays.
We give polynomial-time algorithms for \drc{} (\cref{ch:drc}) and \drg{} (\cref{sec:drg_poly_alg}),
but prove PSPACE-completeness for \drpg{} (\cref{sec:pspacehardness,sec:drpg_pspace_hard}), the variant of the second problem in which no vertex may be visited twice.

\section{Preliminaries}
\label{sec:prelims}
We abbreviate~$\{1, 2, \dots, n\}$ as $[n]$ and $\{n, n+1, \dots, m\}$ as $[n, m]$.
The Iverson bracket~$[P]$ is~$1$ if property~$P$ holds and~$0$~otherwise.
For a time arc~$e = (v, w, t, \lambda_e)$,
we denote the starting and ending vertices as $\tastart(e) = v$ and $\taend(e) = w$, the time label as $\tatime(e) = t$, and the traversal time as $\tatrav(e) = \lambda_e$.
For a vertex~$v$, $\tau_v$~denotes the set of time steps where $v$~has incoming or outgoing time arcs.

\subparagraph{Delays.}
When a time arc~$e$ gets delayed,
then its traversal time~$\lambda(e)$ is increased by~some fixed amount~$\delta$.
For a given set~$D \subseteq E$ of \emph{delayed arcs},
a sequence of time arcs $(v_i, w_i, t_i, \lambda_i)_{i=1}^\ell$ is called
a \emph{$D$-delayed temporal walk}
if it is a temporal walk in the temporal graph obtained from~$\GG$
by applying delays to all time arcs in~$D$.
(We omit~$D$ when it is clear from the context.)

As an example consider the temporal walk $(a, b, 1, 1), (b, c, 3, 1$):
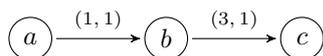
\begin{center}
\begin{tikzpicture}
	\begin{scope}[every node/.style=vertex]
		\node at (0,0) (a) {$a$};
		\node[position=0:8ex from a] (b) {$b$};
		\node[position=0:8ex from b] (c) {$c$};
	\end{scope}
	\draw[timearc]
		(a) edge[edge label={$(1,1)$}] (b)
		(b) edge[edge label={$(3,1)$}] (c);
\end{tikzpicture}
\end{center}
When delaying the first time arc by~1, i.e.\ having $\delta = 1$ and $D = \{(a, b, 1, 1)\}$,
then this is also a delayed temporal walk:
Due to the delay, the first time arc arrives in $b$ at time step $2+\delta=3$ which is still not later than the departure of the second time arc. 
However, if we instead set~$\delta = 2$ and 
$D = \{(a, b, 1, 1)\}$,
then it is no longer a delayed temporal walk,
because the first time arc only reaches~$b$ at time~$4$.

Clearly, from any temporal walk one can obtain a temporal path by eliminating all circular subwalks.
Thus, for any delay set~$D$, if there is a $D$-delayed temporal walk from~$s$ to~$z$,
then there is also a $D$-delayed temporal path.
This is the reason why we did not define a separate version of \drc{} for temporal walks.

Note that this equivalence does not extend to \drg{} as \cref{fig:example_drg_drpg} proves.

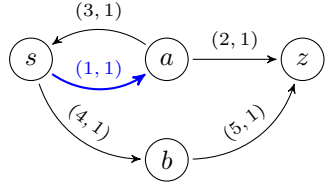
\begin{figure}[t]
	\centering
	\begin{tikzpicture}
		\begin{scope}[every node/.style=vertex]
			\node at (0,0) (s) {$s$};
			\node[position=0:8ex from s] (v1) {$a$};
			\node[below = 5ex of v1] (v2) {$b$};
			\node[right = 8ex of v1] (z) {$z$};
		\end{scope}

		\draw[timearc, thick, blue] (s)[bend left=-35] edge[edge label={$(1,1)$}] (v1);
		\draw[timearc]
			(v1)[] edge[edge label={$(2,1)$}] (z)
			(v1)[bend left=-35] edge[edge label={$(3,1)$}] (s)
			(s)[bend left=-35] edge[edge label={$(4,1)$}] (v2)
			(v2)[bend left=-35] edge[edge label={$(5,1)$}] (z)
			;
    \end{tikzpicture}   
	\caption{
		A temporal graph which, for $x = \delta = 1$, forms a yes-instance of \drg{} but a no-instance of \drpg{}.
		In the former, a winning strategy for the traveler is to take the thick blue time arc to $a$. 
		If it is not delayed, then they can directly go to $z$.
		If it is delayed, then there are no remaining delays.
		Thus, after returning to $s$, the target vertex $z$ can be reached via $b$.
	}
    \label{fig:example_drg_drpg}
\end{figure}

\subparagraph{Static expansion.}
Sometimes, a problem on a temporal graph $\GG$ is transformed to problems on a non-temporal ``time-expanded'' graph,
a \emph{static expansion} of~$\GG$.
The idea is to model each temporal occurrence of every vertex in the temporal graph as a distinct vertex of the static expansion.

Formally, we say that a digraph~$H = (W, A)$ is a \emph{static expansion} of the temporal graph~$\GG = (V, E)$ if
\begin{enumerate}[(i)]
\item $\displaystyle\bigcup_{v \in V}(\{v\} \times \tau_v) \subseteq W \subseteq V \times \NN$, and
\item $A = A_1 \cup A_2$ with
	\begin{align*}
		A_1 &= \left\{ (v, t), (v, t') \mvert (v, t) \in W \land t' = \argmin_{t'' > t}\{ (v, t'') \in W\} \right\}, \\
		A_2 &= \left\{ \left((v, t), (w, t + \lambda)\right) \mvert (v, w, t, \lambda) \in E \right\}.
	\end{align*}
\end{enumerate}
The arcs in $A_1$ are often called \emph{waiting arcs}
while the arcs in $A_2$ are in one-to-one correspondence to the time arcs of~$\GG$.
The static expansion with the minimal set of vertices is called the \emph{reduced static expansion}.

The main virtue of static expansions is that they model temporal walks as (non-temporal) paths.
More precisely, we have the following lemma.
\begin{lemma}
	\label{lemma:static_exp_walk_preserving}
	If $(v, t)$ and $(w, u)$ are two vertices of a static expansion~$H$ of~$\GG$,
	then there is a path from $(v, t)$ to $(w, u)$ if and only if~$\GG$ contains a temporal walk from~$v$ to~$w$
	which starts at time~$t$ or later and arrives at time~$u$ or earlier.
\end{lemma}
The proof of \cref{lemma:static_exp_walk_preserving} is folklore; we omit it, as well as the proof of the following easy observation.
\begin{observation}
	\label{lemma:static_exp_deletion}
	If $H = (W, A)$ is a static expansion of~$\GG = (V, E)$,
	$E' \subseteq E$ is a set of time arcs and $A' \subseteq A$ the set of arcs corresponding to $E'$,
	then $(W, A \setminus A')$ is a static expansion of~$(V, E \setminus E')$.
\end{observation}

\section{\drc{}} \label{ch:drc}

In this section, we present an algorithm (\cref{alg:delay_robust_connection}) which solves \drc{} in polynomial time.
The core idea is to reduce the problem to the computation of a maximum flow problem in a static expansion.
There are three steps in this algorithm. 
First, we construct a new temporal graph~$\GG^*$ in which the removal of a time arc 
has the same effect as delaying the respective time arc in the original input graph~$\GG = (V,E)$. 
Second, we construct a static expansion~$H$ of~$\GG^*$.
Finally, we compute the maximum flow from the start vertex~$s$ to the target vertex~$z$ in~$H$.
We will show that the value of this flow equals the number of delays required to break temporal connectivity between~$s$~and~$z$ in~$\GG$.

For the first step, define $\GG^* = (V, E \cup E^*)$
where
$E^* = \{(v, w, t, \lambda + \delta) \mid (v, w, t, \lambda) \in E\}$. 
Then we can observe the following property of~$\GG^*$.

\begin{lemma} \label{lemma:delay_temp_graph_walk_equiv}
	Let $\GG=(V,E)$ be a temporal graph, $s, z \in V$, $D \subseteq E$,
	and $\GG^* = (V, E \cup E^*)$ as defined above.
	There is a $D$-delayed temporal $(s,z)$-walk in~$\GG$ if and only if there is a temporal $(s,z)$-walk in~$\GG^*_D := (V, (E \setminus D) \cup E^*)$.
\end{lemma}
\begin{proof}
	~\vspace{-\baselineskip} %
	\proofsubparagraph{($\Rightarrow$):} Let $W = (e_1, e_2, \ldots, e_k)$ be a $D$-delayed walk from $s$ to $z$ in $\GG$
	with $e_i = (v_i, w_i, t_i, \delta_i)$ for $i \in [k]$.
	This means that $v_1 = s$, $w_k = z$, and for all~$\ell \in [k-1]$ it holds that $w_\ell = v_{\ell+1}$ and $t_\ell + \lambda_\ell + [e_\ell \in D] \cdot \delta \le t_{\ell+1}$.
	We construct the temporal walk $\hat{W} = (\hat{e}_1, \hat{e}_2, \ldots, \hat{e}_k)$ in $\GG^*_D$ with 
	$
		\hat{e}_i = (v_i, w_i, t_i, \lambda_i + [e_i \in D]\cdot \delta)
	$
	for $i \in [k]$.
	Then
	$
		t(\hat{e}_\ell) + \lambda(\hat{e_\ell}) \le t(\hat{e}_{\ell+1})
	$
	holds for all $\ell \in [k-1]$,
	thus $\hat{W}$ is a temporal walk from~$s$ to $z$ in $\GG^*_D$.

	\proofsubparagraph{($\Leftarrow$):} Let $W = (e_1, e_2, \ldots, e_k)$ be a temporal $(s,z)$-walk in $\GG^*_D$ with~$e_i = (v_i, w_i, t_i, \lambda_i)$.
	This means that $v_1 = s, w_k = z$ and for all $\ell \in [k-1]$, it holds that $w_\ell = v_{\ell+1}$ and $t_\ell + \lambda_\ell \le t_{\ell+1}$.
	We construct a delayed temporal walk $\hat{W} = (\hat{e}_1, \hat{e}_2, \ldots, \hat{e}_k)$ in $\GG$ for the delay set~$D$ with 
	\[
	\hat{e}_i = 
	\begin{cases}
		e_i, 	&\text{if } e_i \in E \setminus D \\
		(v_i, w_i, t_i, \lambda_i - \delta), &\text{otherwise}
	\end{cases}
	\]
	for $i \in [k]$. 	
	Note that $\hat{e}_i \in E$: 
	If $e_i \notin E \setminus D$, then
	$e_i \in E^*$, and thus $\hat{e}_i = (v_i, w_i, t_i, \lambda_i - \delta) \in E$ by construction of $E^*$.	
	We then have for all $\ell \in [k-1]$ that
	\[
		t(\hat{e}_\ell) + \lambda(\hat{e}_\ell) + [\hat{e}_\ell \in D] \cdot \delta
		\le t(e_\ell) + \lambda(e_\ell)
		\le	t(e_{\ell+1})
		= t(\hat{e}_{\ell+1})
	\]
	holds,
	hence $\hat{W}$ is a $D$-delayed temporal walk in~$\GG$.
\end{proof}

The algorithm for \drc{} is given as pseudo-code in \cref{alg:delay_robust_connection}.
\begin{algorithm}[t]
	\caption{\drc{}} \label{alg:delay_robust_connection}
	\begin{algorithmic}[1]
		\Input{$\GG = (V,E), s,z \in V, x,\delta \in \NN$}
		\Output{Is $(\GG, s,z,x,\delta)$ a yes-instance of \drc{}?}
		\State $H \gets \text{reduced static expansion of } \GG^*$
		\State Define $c: E(H) \to \{1, \infty\}$ by $c(e) = \begin{cases} 1 & \text{if $e$ corresponds to a time arc in~$E$} \\ \infty & \text{otherwise} \end{cases}$
		\State $T \gets \argmax_t\{(z, t) \in V(H)\}$
		\State Compute the value~$f$ of a maximum flow from~$(s,0)$ to~$(z,T)$ in~$H$ with capacities~$c$
		\If{$f \le x$}		
		\State \textbf{return} ``NO''
		\Else 
		\State \textbf{return} ``YES''
		\EndIf
	\end{algorithmic}
\end{algorithm}
See also \cref{fig:temp_graph_w_static_expansion} for an illustration of the graph~$H$.
We now show the correctness of the algorithm.

\begin{lemma} \label{lemma:drc-correctness}
	\cref{alg:delay_robust_connection} solves \drc{}. 
\end{lemma}
\begin{proof}
	By \cref{lemma:delay_temp_graph_walk_equiv}, the given instance of \drc{} is a no-instance if and only if
	there exists a set~$D \subseteq E$ of size~$\abs{D} \leq x$ such that $\GG^*_D$ contains no temporal $(s,z)$-walk.
	By \cref{lemma:static_exp_walk_preserving} and \cref{lemma:static_exp_deletion},
	this is equivalent of there being a set~$D' \subseteq E(H)$ of edges in the static expansion~$H$~of~$\GG^*$
	such that
	\begin{enumerate}[(i)]
	\item \label{prop1} $\abs{D'} \leq x$ and the edges in~$D'$ all correspond to edges in~$E$, and
	\item \label{prop2} $H - D'$ contains no walk from $(s, 0)$ to~$(z, T)$ 
	where $T$ is the largest integer for which $H$ contains the vertex~$(z, T)$.
	\end{enumerate}
	Note that \ref{prop2} is equivalent to $D'$ forming a cut set that separates~$(s, 0)$ and $(z, T)$.
	Also, by definition of the capacity function~$c$, \ref{prop1} is equivalent to the total capacity $\sum_{e \in D'} c(e)$ being at most~$x$.
	
	Therefore, by the Max-Flow-Min-Cut-Theorem \cite{MaxFlowMinCut1ford1956maximal,MaxFlowMinCut2:journals/tit/EliasFS56},
	the given instance is a no-instance if and only if the maximum flow from~$(s,0)$ to $(z, T)$ in the graph~$H$ with edge capacities~$c$ is at most~$x$.
\end{proof}

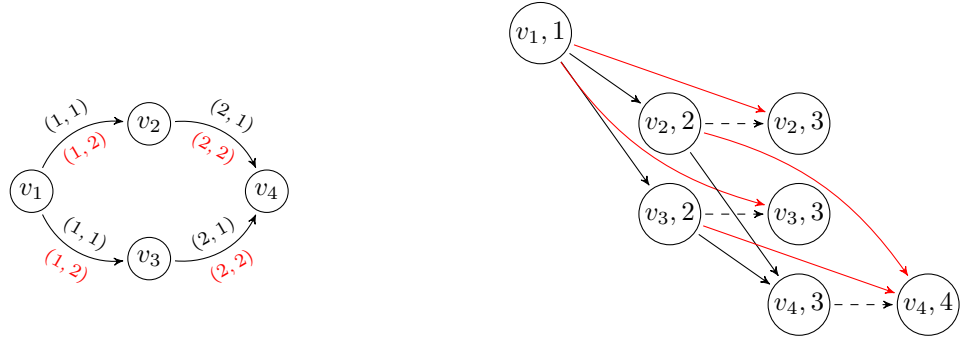
\begin{figure}[t]
	\begin{subfigure}[t]{.45\textwidth}
		\centering
		\begin{tikzpicture}
				\begin{scope}[every node/.style=vertex]
				\node at (0,0) (v1) {$v_1$};
				\node[position=30:8ex from v1] (v2) {$v_2$};
				\node[position=-30:8ex from v1] (v3) {$v_3$};
				\node[position=-30:8ex from v2] (v4) {$v_4$};
			\end{scope}
			
			\draw[timearc] (v1)[bend left=35] edge[edge label={$(1,1)$}, edge label'={\color{red}$(1,2)$}] (v2);
			\draw[timearc] (v1)[bend left=-35] edge[edge label={$(1,1)$}, edge label'={\color{red}$(1,2)$}] (v3);
			\draw[timearc] (v2)[bend left=35] edge[edge label={$(2,1)$}, edge label'={\color{red}$(2,2)$}] (v4);	
			\draw[timearc] (v3)[bend left=-35] edge[edge label={$(2,1)$}, edge label'={\color{red}$(2,2)$}] (v4);
			
			\node[] at (0,-1.8) (invisible) {}; %
		\end{tikzpicture}
		\caption{The temporal graph~$\GG^*$. The black part of the figure shows~$\GG$, while the red time arcs are in~$E^*$.}
	\end{subfigure}  
	\hfill
	\begin{subfigure}[t]{.45\textwidth}
		\centering
		\begin{tikzpicture}[yscale=1]
		
			\begin{scope}[every node/.style=vertex]
			\node at (0.0, 0.0) (v11) {$v_{1},1$};
			\node at (1.7, -1.2) (v22) {$v_{2},2$};
			\node at (3.4, -1.2) (v23) {$v_{2},3$};
			\node at (1.7, -2.4) (v32) {$v_{3},2$};
			\node at (3.4, -2.4) (v33) {$v_{3},3$};
			\node at (3.4, -3.6) (v43) {$v_{4},3$};
			\node at (5.1, -3.6) (v44) {$v_{4},4$};
			\end{scope}

			\draw[timearc] (v11) edge (v22);
			\draw[timearc] (v11) edge (v32);			
			
			\draw[timearc,color=red] (v11) edge (v23);
			\draw[timearc,color=red] [bend left=-20] (v11) edge (v33);	
			
			\draw[timearc,dashed] (v22) edge (v23);
			\draw[timearc,dashed] (v32) edge (v33);	
			
			\draw[timearc] (v22) edge (v43);
			\draw[timearc] (v32) edge (v43);
			
			\draw[timearc,color=red] [bend left=20] (v22) edge (v44);
			\draw[timearc,color=red] (v32) edge (v44);
			
			\draw[timearc,dashed] (v43) edge (v44);

		\end{tikzpicture}
		\caption{The reduced static expansion of~$\GG^*$. The dashed arcs are waiting arcs and the black part forms a static expansion of~$\GG$.
		Note that the capacity function~$c$ assigns~$1$ exactly to the solid black edges.}
	\end{subfigure}
	\caption{Example depiction of $\GG^*$ and its reduced static expansion ($\delta=1$).
}
	\label{fig:temp_graph_w_static_expansion}
\end{figure}

Next, we analyze \cref{alg:delay_robust_connection}'s running time.

\begin{lemma}\label{lemma:drc_running_time}
	\cref{alg:delay_robust_connection} has a running time of $\bigO(\abs{E}^2)$.
\end{lemma}
\begin{proof}
	In the first step, the algorithm constructs the reduced static expansion $H$ of the temporal graph $\GG^*$.
	The size of~$H$ is in~$\bigO(\abs{E})$ and the construction can be done in time linear to its size.
	Constructing~$c$ and~$T$ can also be done in $\bigO(\abs{E})$~time.
	
	Next, we need to compute the value of a maximum flow in~$H$.
	Actually, it suffices to only test whether that value exceeds~$x$ (and we may assume $x \leq \abs{E}$).
	This test is possible in~$\bigO(\abs{E}^2)$~time,
	for example by using the classic method of Ford \& Fulkerson~\cite{MaxFlowMinCut1ford1956maximal}.
\end{proof}

Finally, \cref{lemma:drc-correctness,lemma:drc_running_time} give us the following theorem.

\begin{theorem}
	\drc{} can be solved in quadratic time.
\end{theorem}

\section{Delayed-Routing Games} \label{ch:drgames}

In this section, we analyze the problems \drg{} and \drpg{}, which ask whether a \emph{traveler} can reach their destination when an \emph{adversary} can delay time arcs while the traveler is traversing them.
More formally, the traveler and the adversary are players in a given game instance of the two-player game \drgame{} or \drpgame{}, respectively, and we ask whether the traveler has a winning strategy.
Starting at a start vertex~$s$ at time step~$0$, the traveler selects an out-going time arc from the current vertex, while the adversary can then delay a selected time arc by $\delta$ time units.
However, the number of delays of the adversary is limited to~$x$, thus they cannot always apply a delay.
The traveler wins when they reach the target vertex~$z$.
In the \drpgame{} the traveler can visit each vertex at most once, whereas no such restriction applies in the \drgame{}.

We present a dynamic program to solve \drg{} in $O(\abs{V}\cdot \abs{E}\cdot  x)$ time.
We later use a slightly modified version of the algorithm to show that \drpg{} is in \CPSPACE{}.
Furthermore, using a polynomial-time many-one reduction from \qbfg{} we prove that \drpg{} is \CPSPACE{}-hard. 
Hence, we can conclude that \drpg{} is \CPSPACE{}-complete.

\subsection
	{A dynamic program for \drg{}} \label{sec:drg_poly_alg}

A key observation for our dynamic program is that there are only polynomially many game states in \drg{}.
Furthermore, the available moves from a node of the game tree and the determination of the winner only depend on the current node/game state and not on its predecessors.
Hence, once we have computed whether the traveler has a winning strategy for any given game state, we can save this information in a dynamic programming table.

Let $I = (\GG = (V,E), s,z \in V, \delta,x \in \NN)$ be an instance of \drg{}.
Starting at vertex $s$ at time step~$0$ and with the adversary having a budget of $x$~delays, the goal for the traveler is to reach the target vertex~$z$.
On the traveler's turn, they can select a time arc incident to the current vertex that occurs at the current time step or later.
The adversary can then decide whether they delay this time arc by $\delta$, thus reducing their number of remaining delays by~$1$.
Once the current vertex is the target vertex~$z$, the game ends with the traveler as the winner.
If at any point there are no available time arcs, then the game ends with the adversary as the winner.
If the game runs indefinitely, then the adversary also wins the game
(since~$\GG$ is finite, this can only occur if the traveler follows a cycle with traversal time~$0$).

At the traveler's turn, the game state can be fully described by the $3$-tuple $(v,t,y)$,
where $v \in V$~is the current vertex, $t \in \NN$~is the current time step, and $y\in[0,x]$ is the number of remaining delays.
We may take~$t$~to be from the set $T :=  \{1, \tatime(e) + \tatrav(e), \tatime(e) + \tatrav(e) + \delta \mid e \in E\}$.
The starting game state is $(s,1,x)$.

We define a dynamic programming table $F:V\times T \times[-1,x]\rightarrow \{\texttt{true},\texttt{false}\}$
where $F(v, t, y) = \true$ if the traveler has a winning strategy from the game state~$(v,t,y)$.
Note that we allow the delay budget to reach~$-1$ for technical reasons,
but we will define $F(v, t, -1) = \true$ for all $v \in V$, $t \in T$.
This can be interpreted as allowing the adversary to ``cheat'' by exceeding their budget of delays, at the cost of immediately losing.
Since this option is never beneficial for the adversary, providing it does not change the game in any significant way.

Denote by $E_t(v) := \{(v, w, t', \lambda) \in E \mid t' \geq t\}$
the set of all time arcs that are available at $v \in V$ at time~$t$ or later.
Then we have the following.

\begin{lemma}\label{lem:DP}
For all $v \in V \setminus \{z\}$, $t \in T$, and $y \in [0,x]$
it holds that
\begin{align}
F(z,t,y)&=\true \label{rec:base}\\
F(v,t,-1) &= \true \label{rec:minus}\\
F(v,t,y) &= \bigvee_{\mathclap{e \in E_t(v)}} \big( F(\tadest(e), \tatime(e) + \tatrav(e),y) \land F(\tadest(e), \tatime(e) + \tatrav(e) + \delta,y-1) \big) \label{rec:general}
\end{align}
where the empty disjunction evaluates to \false{}.
\end{lemma}
\begin{proof}
Equation~\eqref{rec:base} is trivially correct, since the traveler has reached their destination vertex~$z$.
Equation~\eqref{rec:minus} holds by definition as noted above.

It remains to prove~\eqref{rec:general}.
By the rules of the game, the traveler may choose any time arc $e \in E_t(v)$ when they are in game state~$(v, t, y)$.
If the adversary opts to delay that time arc, then the resulting game state is
$F(\tadest(e), \tatime(e) + \tatrav(e) + \delta,y-1)$.
Otherwise, the resulting game state is
$F(\tadest(e), \tatime(e) + \tatrav(e),y)$.
If, for any~$e \in E_t(v)$, the traveler has winning strategies for both of these game states, then they can win from $(v, t, y)$ by picking~$e$
and then proceeding from either of the two resulting game states according to their respective winning strategy.
Conversely, if all of the time arcs in~$E_t(v)$ lead to a game state from which the adversary has a winning strategy,
then the traveler clearly cannot win.
\end{proof}

In principle, we would like to use \cref{lem:DP} to compute all value of~$F$.
However, in the presence of arcs with traversal time zero, \cref{lem:DP} might not suffice to completely determine all values of~$F$.
Consider the example instance given in \cref{fig:hung-game}.
For all~$v \in V$, $t \in T$, and~$y \in \{0, 1\}$ we clearly have
	\[
	F(b, t, y) = [t \leq 2] \quad\text{and}\quad
	F(a, t, 0) = [t \leq 2]
	\]
by \eqref{rec:general}.
Note that we can compute
\[
	F(a, 2, 0) = F(b, 2, 0) \lor F(s, 2, 0) = \true \lor F(s, 2, 0) = \true
\]
without needing the value of~$F(s, 2, 0)$.
However, $F(s, 1, 1)$ and $F(a, 1, 1)$ depend on each other through \eqref{rec:general}:
\begin{align*}
	F(s, 1, 1) &= \big(F(a, 1, 1) \land F(a, 2, 0) \big) \lor \big( F(a, 2, 1) \land F(a, 3, 0) \big) = F(a, 1, 1) \quad\text{and}\\
	F(a, 1, 1) &= \big(F(s, 1, 1) \land F(s, 2, 0) \big) \lor \big( F(s, 2, 1) \land F(s, 3, 0) \big) \lor \big( F(b, 2, 1) \land F(b, 3, 0) \big)\\
	&= F(s, 1, 1).
\end{align*}

If the value of a table entry is not determined through \eqref{rec:base}--\eqref{rec:general}, then we call this value \emph{hung}.
As we have seen, it can occur that a subset of all table entries remains hung, even when all other entries have been computed.
This is precisely the case when in each clause appearing in the disjunction~\eqref{rec:general}
at least one of the two referenced entries is is either \false{} or hung itself.
In other words, from a hung state, the traveler only has the options to move to a losing state or to another hung state.
In particular, the traveler can play on forever but never reach a winning state.
Thus, in accordance with our rule that the adversary shall win if the game continues forever, we may set all hung table entries to \false{} in this case.

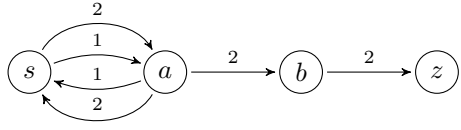
\begin{figure}
	\centering
	\begin{tikzpicture}[scale=1.2]
		\begin{scope}[every node/.style=vertex]
			\node at (0,0) (s) {$s$};
			\node[position=0:8ex from s] (a) {$a$};
			\node[position=0:8ex from a] (b) {$b$};
			\node[position=0:8ex from b] (z) {$z$};
		\end{scope}
		
		\draw[timearc]
			(s) edge[bend left=20, edge label=$1$] (a)
			(s) edge [bend left=60, edge label=$2$] (a)
			(a) edge[bend left=20, edge label=$1$] (s)
			(a) edge[bend left=60, edge label=$2$] (s)
			(a) edge[edge label=$2$] (b)
			(b) edge[edge label=$2$] (z)
			;
	\end{tikzpicture}
	\caption{A \drgame{} instance in which the traveler cycles between~$s$ and~$a$ forever. All traversal times are~$0$, and $\delta = x = 1$.}
	\label{fig:hung-game}
\end{figure}

We summarize this in the following lemma.
\begin{lemma}\label{lem:DPhung}
If every value of~$F$ has either been computed or depends, through~\eqref{rec:general}, on another still uncomputed value,
then the traveler does not have a winning strategy from any of the game states whose values are still uncomputed.
\end{lemma}

Next, we determine the time required to fill the dynamic programming table.
\begin{lemma}\label{lem:DPtime}
All entires of~$F$ can be computed in $O(\abs{V}\cdot \abs{E} \cdot x)$ time.
\end{lemma}
\begin{proof}
The number of table entries is
$N := \abs{V}\cdot \abs{T}\cdot (x+2) \in \bigO(\abs{V} \cdot \abs{E} \cdot x)$.
For any $t \in T$, denote by $t^+$ the smallest element of~$T$ strictly larger than~$t$.
Begin by observing that~\eqref{rec:general} can be replaced with the following equivalent formula.
\begin{align}
\begin{split}
F(v,t,y) ={} &F(v, t^+, y) \lor{} \\
&\bigvee_{\mathclap{e \in E_t(v) \setminus E_{t^+}(v)}}\; \big( F(\tadest(e), \tatime(e) + \tatrav(e),y) \land F(\tadest(e), \tatime(e) + \tatrav(e) + \delta,y-1) \big)  \label{rec:general2}
\end{split}
\end{align}

We compute the table entries using \cref{lem:DP} as follows.
Start with the entries given by~\eqref{rec:base}.
Whenever an entry~$F(v, t, y)$ is set to~$\true$,
check whether any of the entries directly depending on $F(v, t, y)$ through \eqref{rec:general2} can be also computed (i.e., set to \true, as \eqref{rec:general2} contains no negations).
Note that there are three ways of how another entry $F(v', t', y')$ can directly depend on~$F(v, t, y)$:
\begin{enumerate}[(i)]
\item $v = v'$, $y = y'$, and $t = t'^+$,
\item $y' = y$ and there is a time arc $(v', v, \hat{t}, \lambda)$ with $t' \leq \hat{t} < t'^+$ and $\hat{t} + \lambda = t$, or
\item $y' = y + 1$ and there is a time arc $(v', v, \hat{t}, \lambda)$ with $t' \leq \hat{t} < t'^+$ and $\hat{t} + \lambda + \delta = t$.
\end{enumerate}
In particular, each time arc causes a direct dependency only between about $2x$ pairs of entries through (ii) and (iii).
The number of dependencies through (i) is clearly at most~$N$.
Thus, the overall number of checks to be performed is at most $2x \cdot \abs{E} + N$
and each of these checks can be done in constant time.

Afterwards, all remaining entries must be either \false{} or hung:
Since \eqref{rec:general2} contains no negations, setting entries to \false{} can never cause other entries to become \true{}.
By \cref{lem:DPhung}, we can thus set all remaining entries to \false{}.
Hence, the entire table can be filled in $\bigO(x \cdot \abs{E} + N) \subseteq \bigO(\abs{V} \cdot \abs{E} \cdot x)$~time.
\end{proof}

To solve \drg{}, we can now evaluate $F$ and check whether $F(s,1,x)=\texttt{true}$. Hence, \cref{lem:DPtime} gives us the following.
\begin{theorem}
\drg{} can be solved in $O(\abs{V}\cdot \abs{E} \cdot x)$ time.
\end{theorem}

\subsection{PSPACE-hardness of \textsc{Delayed-Routing Path Game}}
\label{sec:pspacehardness}

We now present a polynomial-time reduction from the PSPACE-complete 
\qbfg{} to \drpg{}.
\qbfg{} is a game formulation of the problem \qbf{} that asks whether a given quantified boolean formula is true.
In the game variant, Player~1 and Player~2 choose truth assignments for existentially and universally quantified variables, respectively.
Player~1 wins when the formula is satisfied, otherwise Player~2 wins.
If Player~1 has a winning strategy, then it is a yes-instance.
\qbf{} and \qbfg{} are equivalent and known to be \CPSPACE{}-complete \cite{DBLP:qbfPSPACEcomp}.

In \qbfg{}, we are given a quantified boolean formula $\Phi = Q_1 x_1. Q_2 x_2. \ldots Q_n x_n. \varphi$ with $Q_i \in \{\exists, \forall\}$ and $\varphi$ being a boolean formula.
The game then consists of $n$ rounds.
In the $i$-th round, if $Q_1 = \exists$, then Player~1 selects a truth value for $x_i$.
Else if $Q_1 = \forall$, then Player~2 selects a truth value for $x_i$.
If after the $n$-th round $\varphi$ is satisfied under the selected truth assignment, then Player~1 wins, otherwise Player~2 is the winner.

\problemdef{QBF Game}
{A quantified boolean formula $\Phi = Q_1 x_1. Q_2 x_2. \ldots Q_n x_n. \varphi$ with $Q_i \in \{\exists, \forall\}$.} %
{Is there a winning strategy for Player~1?}

Given a \qbfg{}-instance $\Phi = Q_1 x_1. Q_2 x_2. \ldots Q_n x_n. \varphi$, we construct an instance $(\GG = (V,E), s,z\in V, x, \delta \in \NN)$ of \drpg{}, so that Player~1 has a winning strategy for $\Phi$ if and only if the traveler has a winning strategy for the \drpg{}-instance.
Without loss of generality, we assume that~$\varphi$ is in conjunctive normal form and has three literals per clause.

\subsubsection{Description of the reduction}
The main idea for our reduction is the following. The temporal graph $\GG$ we create in our reduction consists of $n$ chained selection gadgets, one for each quantified variable, and $m-1$ validation gadgets
where $m$~is the number of clauses of~$\varphi$.
In a selection gadget for a variable quantified by an existential quantifier, the traveler can freely choose one of two paths, corresponding to a truth assignment of this variable. 
A delay by the adversary has no effect in this gadget.
In a selection gadget for a variable quantified by a universal quantifier, a delay of the adversary forces the traveler to take a specific path, corresponding to a truth assignment of this variable. 
If not taking this enforced path, the traveler gets immediately stuck and loses the game.
The gadget is constructed in a way that the adversary needs to use exactly one delay.
Using no delay lets the traveler immediately win, while using more than one delay is no better for the adversary than a single delay.

In the validation gadgets, the adversary can force the traveler to take one of two paths, one corresponding to selecting the corresponding clause of $\varphi$, the other leading to the next validation gadget.
In this way, the adversary can select one clause of $\varphi$.
After traversing the validation gadgets, the adversary has no remaining delays.

Finally, if there is a literal in the clause which is satisfied under the truth assignment corresponding to the path taken in the selection gadgets, then the traveler can traverse back to a vertex in the selection gadget.
From this vertex, the traveler can then reach the target vertex.
Otherwise, if all literals in the clause are unsatisfied, then all vertices reachable with a time arc have already been visited when traversing the selection gadgets.
Thus the traveler becomes stuck.

Now we describe the reduction more formally. 
All time arcs in the constructed temporal graph $\GG$ have a traversal time of $0$, thus we write time arcs as $3$-tuples $(v,w,t) \in E$ and omit the traversal time in figures.
We set the number of delays $x := n' + m - 1$, where $n' \le n$ is the number of universal quantifiers in~$\Phi$.
Furthermore, we set $\delta := 1$.
The start and target vertices of the game are $s_1$ and~$z$, respectively, which are added during the construction of the temporal graph.
The gadgets use an offset $o_i$, starting with $o_1 = 0$.
The other offsets are computed while constructing the gadgets.
Initially, we add the vertices $s_1, s_2, \ldots, s_{n+1}$, $z'$, and $z$ to $V$, and we add the time arc 
\[ z' \xrightarrow{o_{n+m}+1} z. \]
The time step $o_{n+m}+1$ is the largest time step of the constructed temporal graph.

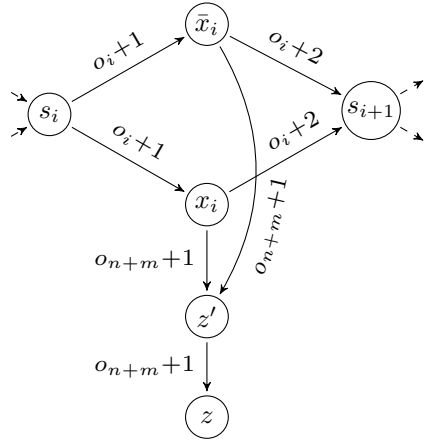
\begin{figure} [t]
	\centering
	\tikzstyle{alter}=[circle, minimum size=16pt, draw, inner sep=1pt] 
	\tikzstyle{majarr}=[draw=black]
	
	\begin{tikzpicture}	[auto, >=stealth',shorten <=1pt, shorten >=1pt]
		\tikzstyle{majarr}=[draw=black,->,shorten <=1.5pt, shorten >=1.5pt]
		\node[alter] at (0,0) (s) {$s_i$};
		\node[alter, position= -30:12ex from s] (x) {$x_i$};
		\node[alter, position= 30:12ex from s] (nx) {$\bar{x}_i$};
		\node[alter, position= 30:12ex from x] (ss) {$s_{i+1}$};
		
		\draw[majarr] (s) edge node[scale=1.25,sloped, midway, anchor=south]{$\scriptstyle o_i + 1$} (x);
		\draw[majarr] (s) edge node[scale=1.25,sloped, midway, anchor=south]{$\scriptstyle o_i + 1$} (nx);		
		\draw[majarr] (x) edge node[scale=1.25,sloped, midway, pos=0.65, anchor=south]{$\scriptstyle o_i + 2$} (ss);
		\draw[majarr] (nx) edge node[scale=1.25,sloped, midway, anchor=south]{$\scriptstyle o_i + 2$} (ss);

		\node[alter, below=6ex of x]  (zp) {$z'$};
		\node[alter, below=5ex of zp] (z) {$z$};

		\draw[majarr] (zp) edge 
			node[scale=1.25,midway, anchor=east]{$\scriptstyle o_{n+m}+1$} 
			(z);				
		\draw[majarr] (x) edge 
			node[scale=1.25,midway, anchor=east]{$\scriptstyle o_{n+m}+1$} 
			(zp);
		\draw[majarr] (nx) edge[bend left=30] 
			node[scale=1.25,sloped, pos=0.7, anchor=north]{$\scriptstyle o_{n+m}+1$} 
			(zp);

		\node[position= -30:-8ex from s] (pd1) {};
		\node[position= 30:-8ex from s] (pd2) {};

		\node[position= -30:3ex from ss] (sc1) {};
		\node[position= 30:3ex from ss] (sc2) {};

		\draw[majarr, dashed] (pd1) edge (s);
		\draw[majarr, dashed] (pd2) edge (s);
		\draw[majarr, dashed] (ss) edge (sc1);
		\draw[majarr, dashed] (ss) edge (sc2);

	\end{tikzpicture}

	\caption{Selection gadget for the existentially quantified variable $x_i$.
	The traveler can choose freely whether the upper or lower path to $s_{i+1}$ is taken. 
	The adversary has no incentive to apply any delays.
	Dwelling in $x_i$ or $\bar{x}_i$ to get to $z'$ is no option for the traveler as long as there are delays left.}
	\label{fig:qbf_sel_gadgets_exist}
\end{figure}

\begin{figure} [ht]
	\centering
	\tikzstyle{alter}=[circle, minimum size=16pt, draw, inner sep=1pt] 
	\tikzstyle{majarr}=[draw=black]
	
	\begin{tikzpicture}	[auto, >=stealth',shorten <=1pt, shorten >=1pt]
		\tikzstyle{majarr}=[draw=black,->,shorten <=1.5pt, shorten >=1.5pt]
		
		\node[alter] at (0,0) (ss) {$s_{i}$};
		
		\node[alter, position= 0:6ex from ss] (ssp) {$s_{i}'$};
		\node[alter, position= 30:9ex from ssp] (nxx) {$\bar{x}_i$};
		\node[alter, position= 0:6ex from nxx] (nxx1) {$\bar{x}_i^{(1)}$};
		\node[position= 0:6ex from nxx1] (nxxi) {\ldots};
		\node[alter, position= 0:12ex from nxxi, scale=.75] (nxxk) {$\bar{x}_i^{(n'+m-i'-2)}$};
		
		\node[position= -40:7ex from ssp] (dummy) {};			

		\node[alter,] at (dummy -| nxxi) (xx) {$x_i$};
		\node[alter] at ([xshift=18ex]ss -| nxxk) (sss) {$s_{i+1}$};

		\draw[majarr, line width=1.5pt, color=blue] (ss) edge node[midway, anchor=south]{$\scriptstyle o_{i} + 1$} (ssp);
		\draw[majarr] (ssp) edge node[sloped, midway, anchor=south]{$\scriptstyle o_{i} + 1$} (xx);
		\draw[majarr] (xx) edge 
			node[sloped, midway, anchor=south]{$\scriptstyle o_{i} + 1$} 
			node[sloped, midway, anchor=north]{$\scriptstyle o_{i} + 2$}
			(sss);		
		
		\draw[majarr] (ssp) edge node[sloped, midway, anchor=south]{$\scriptstyle o_{i} + 2$} (nxx);
		\draw[majarr] (nxx) edge 
			node[midway, anchor=south]{$\scriptstyle o_{i} + 2$}
			node[midway, anchor=north]{$\scriptstyle o_{i} + 3$} 
			(nxx1);
		\draw[majarr] (nxx1) edge 
			node[midway, anchor=south]{$\scriptstyle o_{i} + 3 $} 
			node[midway, anchor=north]{$\scriptstyle o_{i} + 4 $} 
			(nxxi);
		\draw[majarr] (nxxi) edge 
			node[midway, anchor=south]{$\scriptstyle o_{i} + n'+m-i'-1$} 
			node[midway, anchor=north]{$\scriptstyle o_{i} + n'+m-i'$} 
			(nxxk);
		\draw[majarr] (nxxk) edge 
			node[sloped, midway, anchor=south]{$\scriptstyle o_{i} + n'+m-i'$} 
			(sss);

		\node[alter] at ([yshift=-9ex]dummy -| nxx1) (zp) {$z'$};
		\node[alter, below=5ex of zp] (z) {$z$};

		\draw[majarr] (zp) edge 
			node[midway, anchor=east]{$\scriptstyle o_{n+m}+1$} 
			(z);				
		\draw[majarr] (xx) edge 
			node[sloped, midway, anchor=north]{$\scriptstyle o_{n+m}+1$} 
			(zp);
		\draw[majarr] (nxx) edge[bend right=50] 
			node[sloped, midway, anchor=north]{$\scriptstyle o_{n+m}+1$} 
			(zp);

		\draw[majarr] (sss) edge 
			node[sloped, midway, anchor=north]{$\scriptstyle o_{i} + 1$} 
			(z);

		\node[position= -30:-8ex from ss] (pd1) {};
		\node[position= 30:-8ex from ss] (pd2) {};

		\node[position= -30:3ex from sss] (sc1) {};
		\node[position= 30:3ex from sss] (sc2) {};

		\draw[majarr, dashed] (pd1) edge (ss);
		\draw[majarr, dashed] (pd2) edge (ss);
		\draw[majarr, dashed] (sss) edge (sc1);
		\draw[majarr, dashed] (sss) edge (sc2);

	\end{tikzpicture}

	\caption{Selection gadget for the universally quantified variable $x_i$. Let $i'$ be the number of universally quantified variables before the $i$-th variable.
	The traveler enters with $n'+m-i'-1$ delays left.
	If the adversary delays the first (thick, blue) time arc, only the upper path to $s_{i+1}$ is available for the traveler.
	The remaining $n'+m-i'-2$ delays are not enough to make the traveler stuck, and the adversary has no incentive to use any further delays.
	If the adversary does not delay the first (thick, blue) time arc, then the traveler is forced to take the lower path to $s_{i+1}$, since if the adversary applies all $n'+m-i'-1$ delays on the upper time arcs, then the traveler gets stuck.
	The adversary will have to apply one delay on the lower time arcs, otherwise the traveler can take the time arc from $s_{i+1}$ to $z$.	
	}
	\label{fig:qbf_sel_gadgets_univ}
\end{figure}
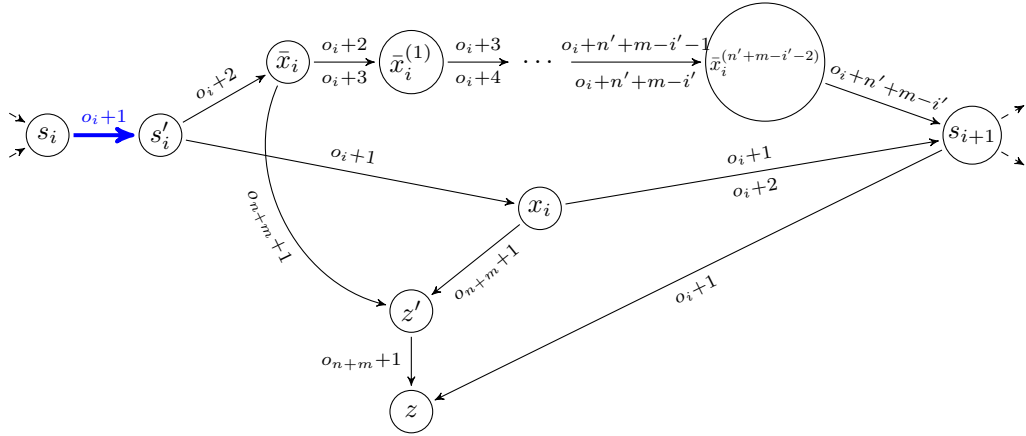

\paragraph*{Selection Gadgets.} The selection gadget is used to assign a truth value to the quantified variable $Q_i x_i$. 
The gadget depends on the type of quantifier: 

\subparagraph*{Case 1.} $Q_i x_i = \exists x_i$; the variable $x_i$ is existentially quantified. 

We add the vertices $x_i$ and $\bar{x}_i$ to $V$.
Furthermore, we add the following time arcs
\[
	s_i \xrightarrow{o_i + 1} x_i \xrightarrow{o_i+2} s_{i+1} 
	\text{\quad and \quad} 
	s_i \xrightarrow{o_i + 1} \bar{x}_i \xrightarrow{o_i+2} s_{i+1}.
\]
corresponding to assigning $x_i$ to true and false, respectively. 
The traveler can choose whether to reach $s_{i+1}$ over the vertex $x_i$ or $\bar{x}_i$.
A delay of the adversary will have no effect.

Additionally, we add the time arcs 
\[
x_i \xrightarrow{o_{m+n}+1} z' \text{\quad and \quad} \bar{x}_i \xrightarrow{o_{m+n}+1} z';
\]
however, when traversing this gadget, if the traveler dwells in $x_i$ or $\bar{x}_i$ to take the time arc to~$z'$, then the adversary can use a delay that makes the only outgoing time arc to $z$ at time step $o_{m+n}+1$ unavailable.

We set the offset $o_{i+1} := o_i + 2$.
An example of a selection gadget for existentially quantified variables can be seen in \cref{fig:qbf_sel_gadgets_exist}.

\subparagraph*{Case 2.} $Q_i x_i = \forall x_i$; the variable $x_i$ is universally quantified. 

Let $i'$ be the number of universally quantified variables before the $i$-th variable.
We add the vertices $s_i', x_i, \bar{x}_i$, and $\bar{x}_i^{(1)}, \bar{x}_i^{(2)}, \ldots, \bar{x}_i^{(n'+m-i'-2)}$ to $V$.
Furthermore, we add a time arc $s_i \xrightarrow{o_i + 1} s_i'$, and the time arcs
\[
s_i' \xrightarrow{o_i + 1} x_i \xrightarrow{o_i + 1, o_i + 2} s_{i+1},
\]
corresponding to setting $x_i$ to true, and 
\[
s_i \xrightarrow{o_i + 2} \bar{x}_i \xrightarrow[o_i+3]{o_i+2} \bar{x}_i^{(1)} \xrightarrow[o_i+4]{o_i+3} \bar{x}_i^{(2)} \ldots \xrightarrow[o_i+n'+m-i']{o_i+n'+m-i'-1} \bar{x}_i^{(n'+m-i'-2)} \xrightarrow{o_i+n'+m-i'} s_{i+1},
\]
corresponding to setting $x_i$ to false.
Finally, we add a time arc $s_{i+1} \xrightarrow{o_i + 1} z$ directly to the end vertex $z$.

By not delaying the time arc $s_i \xrightarrow{o_i + 1} s_i'$, the adversary forces the traveler to take the path through vertex $x_i$, since the path $s_i \rightarrow \bar{x}_i \rightarrow \bar{x}_i^{(1)} \rightarrow \ldots \rightarrow \bar{x}_i^{(n'+m-i'-2)} \rightarrow s_{i+1}$ can be broken by applying all remaining $n'+m-i'-1$ delays.
The adversary is still enforced to apply one delay in $s_i' \rightarrow x_i \rightarrow s_{i+1}$, otherwise the traveler can take $s_{i+1} \xrightarrow{o_i + 1} z$ and directly wins. 
By delaying the time arc $s_i \xrightarrow{o_i + 1} s_i'$, the adversary forces the traveler to take the path through vertex $\bar{x}_i$, since the time arc $s_i \xrightarrow{o_i + 1} x_i$ becomes unavailable.
However, the remaining $n'+m-i'$ delays are not enough to break the path $s_i \rightarrow \bar{x}_i \rightarrow \bar{x}_i^{(1)} \rightarrow \bar{x}_i^{(2)} \ldots \bar{x}_i^{(n'+m-i')} \rightarrow s_{i+1}$.

Additionally, we add the time arcs 
\[
x_i \xrightarrow{o_{m+n}+1} z' \text{\quad and \quad} \bar{x}_i \xrightarrow{o_{m+n}+1} z',
\]
however when traversing this gadget, if the traveler dwells in $x_i$ or $\bar{x}_i$ to take the time arc to $z'$, then the adversary can use a delay that makes the only outgoing time arc to $z$ at time step $o_{m+n}+1$ unavailable.

We set the offset $o_{i+1} := o_i+n'+m-i'$.
An example of a selection gadget for universally quantified variables can be seen in \cref{fig:qbf_sel_gadgets_univ}.

\paragraph*{Validation Gadgets.}
The validation gadgets are used to check whether the formula $\varphi$ is satisfied for the truth assignment chosen in the selection gadgets.
By using all $m-1$ remaining delays, the adversary can force the traveler to visit a vertex corresponding to a specific clause of $\varphi$.
For the clauses $c_1, c_2, \ldots, c_{m-1}$ there is a validation gadget. 
By placing a single delay, the adversary can force the traveler to take one of two junctions, where one corresponds to selecting the clause, and the other leads to the next validation gadget. (For the $m-1$-st validation gadget the other junction corresponds to the $m$-th clause.)
From there, the traveler can reach $z$ only if there is a satisfied literal in the clause.

For each clause $c_i \in \{c_1, c_2, \ldots, c_{m-1}\}$, we add the vertices $v_i$, $v_i^{(l,k)}$ for $k \in [m-i]$, and $v_i^{(r,k)}$ for $k \in [m-i]$.
For $i \in [2, m-1]$, we add the time arc 
$$
	v_{i-1}^{(r, m-(i-1))} \xrightarrow{o_{n+i}+1} v_i,
$$
connecting the $i$-th validation gadget with the $(i-1)$-st validation gadget.
For $i = 1$, we add the time arc
$$
	s_{n+1} \xrightarrow{o_{n+1}+1} v_i,
$$
connecting the last selection gadget with the first validation gadget.

For the branch corresponding to selecting the $i$-th clause, we add the time arcs 
\[
	v_i \xrightarrow{o_{n+i}+1} v_i^{(l,1)} \xrightarrow{o_{n+i}+2} v_i^{(l,2)} \xrightarrow{o_{n+i}+3} \ldots \xrightarrow{o_{n+i}+m-i} v_i^{(l,m-i)}.
\]
Furthermore, for all $v_i^{(l,k)}$ with $k \in [m-i]$, we add a time arc
\[
	v_i^{(l,k)} \xrightarrow{o_{n+i}+k} z,
\]
which enforces the adversary to place delays on all time arcs above, so the traveler cannot directly go to vertex $z$ and win the game.
This ensures that all delays are used after reaching~$v_i^{(l,m-i)}$, which corresponds to selecting the $i$-th clause.

For the branch corresponding to not selecting the $i$-th clause, we add the time arc
\[
v_i \xrightarrow{o_{n+i}+2} v_i^{(r,1)},
\] 
the time arcs
\[
v_i^{(r,k)} \xrightarrow[o_{n+i}+k+2]{o_{n+i}+k+1} v_i^{(r,k+1)},
\]
for $k \in [m-i-1]$, and the time arc
\[
v_i^{(r,m-i)} \xrightarrow{o_{n+i}+m-i+1} v_i^{(r,m-i+1)}.
\]
Delaying all $m-i$ time arcs from $v_i$ to $v_i^{(r,m-i)}$ will break the connection from $v_i^{(r,m-i)} \rightarrow v_i^{(r,m-i+1)}$; however, if there is one delay less, then the adversary does not get any better by using delays.
We set the next offset $o_{n+i+1} := o_{n+i}+m-i+2$.

Finally, we add time arcs for the literals in the clauses.
Let the $i$-th clause be 
	$$(l_i^{(1)} \vee l_i^{(2)} \vee l_i^{(3)}).$$
For all $i \in [m-1]$ and all $j \in [3]$, if $l_i^{(j)} = x_k$, then we add the time arc 
\[
	v_i^{(l,m-i)} \xrightarrow{o_{n+m}+1} \bar{x}_k,
\]
and if $l_i^{(j)} = \bar{x}_k$, then we add the time arc 
\[
	v_i^{(l,m-i)} \xrightarrow{o_{n+m}+1} x_k,
\]
where $k \in [n]$.
For the last clause $c_m$ and for all $j \in [3]$, if $l_i^{(j)} = x_k$, we add the time arc
\[
	v_{m-1}^{(r,2)} \xrightarrow{o_{n+m}+1} \bar{x}_k,
\]
and if $l_i^{(j)} = \bar{x}_k$, then we add the time arc
\[
	v_{m-1}^{(r,2)} \xrightarrow{o_{n+m}+1} x_k,
\]
where $k \in [n]$.
At this point there are no delays remaining, and for all $k \in [n]$, the time arcs
\[
	x_k \xrightarrow{o_{n+m}+1} z' \xrightarrow{o_{n+m}+1} z 
	\text{\quad or \quad} 
	\bar{x}_k \xrightarrow{o_{n+m}+1} z' \xrightarrow{o_{n+m}+1} z,
\]
which have been added previously in the selection gadgets, can be traversed.

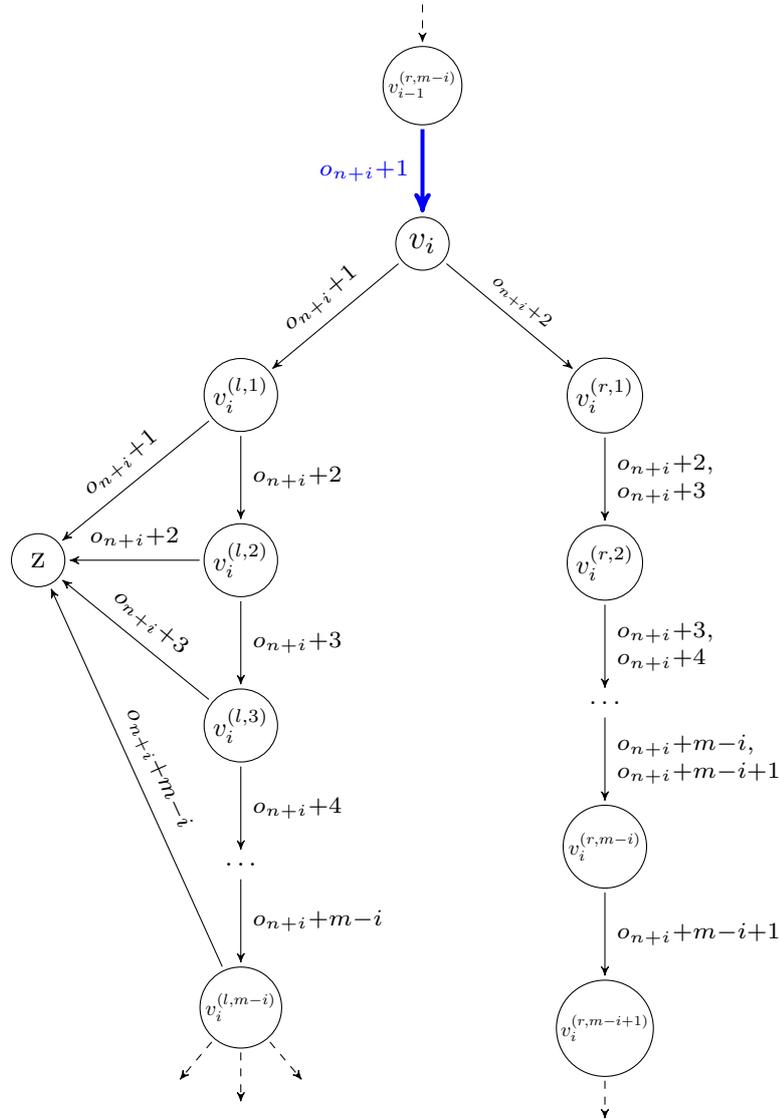
\begin{figure} [htbp]
	\centering
	\tikzstyle{alter}=[circle, minimum size=16pt, draw, inner sep=1pt] 
	\tikzstyle{majarr}=[draw=black]
	
	\begin{tikzpicture}	[auto,>=stealth',shorten <=1pt, shorten >=1pt]
		\tikzstyle{majarr}=[draw=black,->,shorten <=1.5pt, shorten >=1.5pt]
		
		\node[alter, scale=0.75] at (0,0) (vp) {$v_{i-1}^{(r,m-i)}$};
		\node[alter, position= -90:8ex from vp, scale=1.25] (v) {$v_{i}$};

		\node[alter, position= -40:15ex from v] (r1) {$v_{i}^{(r,1)}$};
		\node[alter, position= -90:8ex from r1] (r2) {$v_{i}^{(r,2)}$};
		\node[position= -90:8ex from r2] (r3) {$\ldots$};
		\node[alter, position= -90:8ex from r3, scale=0.8] (r4) {$v_{i}^{(r,m-i)}$};
		\node[alter, position= -90:8ex from r4, scale=0.75] (r5) {$v_{i}^{(r,m-i+1)}$};

		\node[alter, position= -140:15ex from v] (l1) {$v_{i}^{(l,1)}$};
		\node[alter, position= -90:8ex from l1] (l2) {$v_{i}^{(l,2)}$};
		\node[alter, position= -90:8ex from l2] (l3) {$v_{i}^{(l,3)}$};
		\node[position= -90:8ex from l3] (l4) {$\ldots$};
		\node[alter, position= -90:8ex from l4, scale=0.8] (l5) {$v_{i}^{(l,m-i)}$};

		\node[alter, scale=1.25, left=12ex of l2] (z) {z};
		\node[draw=none, right=12ex of r2] (dummyr) {}; %

		\draw[majarr, line width=1.5pt, color=blue] (vp) edge node[scale=1.25,midway, anchor=east]{$\scriptstyle o_{n+i} + 1$} (v);

		\draw[majarr] (v) edge node[scale=1.25,sloped,midway, anchor=south]{$\scriptstyle o_{n+i} + 1$} (l1);
		\draw[majarr] (l1) edge node[scale=1.25,midway, anchor=west]{$\scriptstyle o_{n+i} + 2$} (l2);
		\draw[majarr] (l2) edge node[scale=1.25,midway, anchor=west]{$\scriptstyle o_{n+i} + 3$} (l3);
		\draw[majarr] (l3) edge node[scale=1.25,midway, anchor=west]{$\scriptstyle o_{n+i} + 4$} (l4);
		\draw[majarr] (l4) edge node[scale=1.25,midway, anchor=west]{$\scriptstyle o_{n+i} + m-i$} (l5);

		\draw[majarr] (l1) edge node[scale=1.25,sloped, midway, anchor=south]{$\scriptstyle o_{n+i} + 1$} (z);
		\draw[majarr] (l2) edge node[scale=1.25,sloped, midway, anchor=south]{$\scriptstyle o_{n+i} + 2$} (z);
		\draw[majarr] (l3) edge node[scale=1.25,sloped, midway, anchor=south]{$\scriptstyle o_{n+i} + 3$} (z);
		\draw[majarr] (l5) edge node[scale=1.25,sloped, midway, anchor=south]{$\scriptstyle o_{n+i} + m-i$} (z);

		\draw[majarr] (v) edge node[sloped,midway, anchor=south]{$\scriptstyle o_{n+i} + 2$} (r1);
		\draw[majarr] (r1) edge 
			node[scale=1.25, pos=.35, anchor=west]{$\scriptstyle o_{n+i} + 2,$} 
			node[scale=1.25, pos=.65, anchor=west]{$\scriptstyle o_{n+i} + 3$} (r2);
		\draw[majarr] (r2) edge 
			node[scale=1.25,pos=.35, anchor=west]{$\scriptstyle o_{n+i} + 3,$} 
			node[scale=1.25,pos=.65, anchor=west]{$\scriptstyle o_{n+i} + 4$} (r3);
		\draw[majarr] (r3) edge 
			node[scale=1.25,pos=.35, anchor=west]{$\scriptstyle o_{n+i} + m-i,$} 
			node[scale=1.25,pos=.65, anchor=west]{$\scriptstyle o_{n+i} + m-i+1$} (r4);
		\draw[majarr] (r4) edge 
			node[scale=1.25,midway, anchor=west]{$\scriptstyle o_{n+i} + m-i+1$} (r5);

		\node[draw=none,position=-130:5ex from l5] (lit1) {};
		\node[draw=none,position=-90:5ex from l5] (lit2) {};
		\node[draw=none,position=-50:5ex from l5] (lit3) {};

		\draw[majarr,dashed] (l5) edge (lit1);
		\draw[majarr,dashed] (l5) edge (lit2);
		\draw[majarr,dashed] (l5) edge (lit3);

		\node[draw=none,position=-90:4ex from r5] (succ_val) {};
		\draw[majarr,dashed] (r5) edge (succ_val);

		\node[draw=none, position= 90:4ex from vp] (predummy) {};
		\draw[majarr, dashed] (predummy) edge (vp) {};

	\end{tikzpicture}

	\caption{Validation gadget for the $i$-th clause. 
	The left branch corresponds to selecting the clause, the right branch leads into the next validation gadget. 
	The gadget is entered with $m-i$ delays left. 
	If no delay occurs on the first (thick, blue) time arc, then the traveler is forced to take the left branch, since on the right branch, by using all $m-1$ delays, the traveler gets stuck. The adversary is enforced to use all $m-1$ delays on the left branch, so that the shortcut to $z$ is unavailable.
	If a delay occurs on the first (thick, blue) time arc, then the traveler can only take the right branch.
	In this case, the remaining $m-i-1$ delays are not enough to make the traveler stuck, and the adversary has no incentive to use further delays. 
	}
	\label{fig:qbf_val_gadgets}
\end{figure}
An example validation gadget can be seen in \cref{fig:qbf_val_gadgets}.

\subsubsection{Proof of correctness}

We begin the proof of correctness by first stating the following observation that for the adversary, two game states with the same current vertex and the same set of available time arcs, the game state with more remaining delays is better.

\begin{observation} \label{obs:more_delays_better}
	Let $s_1 = (v,t_1,y_1, V')$ and $s_2 = (v,t_2,y_2, V')$ be two game states, where $v$ is the current vertex, $t_1$ and $t_2$ are the current time steps, $y_1$ and $y_2$ are the remaining delays, and $V'$ is the set of visited vertices.
	Let $t'_1$ and $t'_2$ be the next larger time step from $t_1$ and $t_2$, respectively, where $v$ has out-going time arcs.
	If $t'_1 = t'_2$ and $y_1 \ge y_2$, then all strategies that can be applied in the game state $s_1$ can also be applied in the game state $s_2$.
\end{observation}

The traveler's moves are not dependent on the number of remaining delays.
If the current vertex stays the same, then changing the time step to the next larger time step where this vertex has out-going time arcs does not change the available time arcs / set of moves.
The adversary never has a benefit when having no remaining delays. 
Hence, it is better for the adversary to save the delay for a later use.

For the proof, we will look at \textit{rational strategies} for both players. 
A \textit{rational strategy} is a strategy which does not select moves that let the other player win within the same gadget, and for the adversary, additionally, does not use unnecessary delays (see \cref{obs:more_delays_better}).
Hence, even if the other player has a winning strategy, a rational strategy will delay the loss of the game as far as possible.
If a player has a winning strategy, then they also have a rational strategy.
Hence, we only need to look at rational strategies. 
This will reduce the number of possible moves we have to analyze.

Next, we present an invariant that states the number of delays when entering or leaving the selection gadgets.
\begin{invariant} \label{inv:sel_gad_inv}
	When both players use a rational strategy, then for all $i \in [n+1]$ the vertex~$s_i$ is entered with $n'+m-i'-1$ remaining delays, where $i'$ is the number of variables quantified by a universal quantifier until the $i$-th variable.
\end{invariant}

The correctness of this invariant follows from \cref{obs:qbf_game_sel_gad_inv_start} and \cref{lemma:qbf_game_sel_gad_exist,lemma:qbf_game_sel_gad_univ}. 
We can see that the invariant holds at the start of the game.

\begin{observation} \label{obs:qbf_game_sel_gad_inv_start}
	\cref{inv:sel_gad_inv} holds at the beginning of the game.
\end{observation}
\begin{proof}
	The game starts at $s_1$ with a total of $n'+m-1$ delays.
	There are no variables before the first variable, hence the invariant holds.
\end{proof}

The following lemma shows that the invariant holds after traversing a selection gadget for a variable quantified by an existential quantifier.
Furthermore, it shows that the traveler is in control of which of the two paths through the gadget is taken.

\begin{lemma} \label{lemma:qbf_game_sel_gad_exist}
	Let the $i$-th quantified variable in $\Phi$ be $\exists x_i$.
	When both players use a rational strategy, the traveler can decide whether the vertices $s_i \rightarrow x_i \rightarrow s_{i+1}$ or the vertices $s_i \rightarrow x_i \rightarrow s_{i+1}$ are traversed.
	Additionally, \cref{inv:sel_gad_inv} holds in $s_{i+1}$.
\end{lemma}
\begin{proof}
	Let $i'$ be the number of variables quantified by a universal quantifier before the $i$-th variable. 
	Due to \cref{inv:sel_gad_inv}, there are $n'+m-i'-1$ remaining delays at the start vertex of the selection gadget $s_i$.
	If $i = 1$, then it is the first gadget, and the current time step is $0 = o_1$.
	Otherwise if $i > 1$, then the incoming time arcs in $s_i$ are at latest at $o_i$, hence the current time in~$s_i$ is at latest $o_i+1$.
	Hence, both time arcs $s_i \xrightarrow{o_i+1} x_i$ and $s_i \xrightarrow{o_i+1} \bar{x}_i$ are available.
	Due to \cref{obs:more_delays_better} the adversary will not use any delay, since the next time step where these vertices have out-going time arcs is $o_i+2$.
	From vertex $x_i$ or $\bar{x}_i$ the traveler cannot take the time arc to $z'$ at $o_{n+m}+1$, since the adversary can use a delay, and the only time arc $z' \xrightarrow{o_{n+m}+1} z$ from $z'$ becomes unavailable leading to a loss for the traveler.
	Hence, the traveler will take the time arc to $s_{i+1}$ at $o_i+2$.
	Again due to \cref{obs:more_delays_better} The adversary will not use a delay on this time arc, since the next out-going time arc at $s_{i+1}$ is at $o_{i+1}+1 = o_i+3$.
	Thus, the number of remaining delays is unchanged, and the number of variables quantified by a universal quantifier before the $(i+1)$-st variable is still~$i'$.
	Consequently, \cref{inv:sel_gad_inv} holds at $s_{i+1}$.
\end{proof}

The next lemma shows that the invariant holds after traversing a selection gadget for a variable quantified by a universal quantifier.
Furthermore, it shows that the adversary is in control of which of the two paths through the gadget is taken.

\begin{lemma} \label{lemma:qbf_game_sel_gad_univ}
	Let the $i$-th quantified variable in $\Phi$ be $\forall x_i$.
	When both players use a rational strategy, the adversary can decide whether the vertex $s_{i+1}$ is reached through the vertex~$x_i$ or the vertex~$\bar{x}_i$.
	Additionally, \cref{inv:sel_gad_inv} holds in $s_{i+1}$.
\end{lemma}
\begin{proof}
	Let $i'$ be the number of variables quantified by a universal quantifier before the $i$-th variable. 
	Due to \cref{inv:sel_gad_inv}, there are $n'+m-i'-1$ remaining delays.
	If $i = 1$, then it is the first gadget, and the current time step is $0 = o_1$.
	Otherwise if $i > 1$, then the incoming time arcs in $s_i$ are at latest at $o_i$, hence the current time in $s_i$ are at latest $o_i+1$.
	Hence, the time arc $s_i \xrightarrow{o_i+1} s'_i$ is available.
	Since it is the only outgoing time arc from $s_i$, the traveler has to take it.
	Depending on whether the adversary is delaying this time arc, the adversary can enforce which path to~$s_{i+1}$ is taken by the traveler.
	
	\textit{Case 1:} The adversary uses a delay on $s_i \xrightarrow{o_i+1} s'_i$ and thus the traveler arrives in $s'_i$ at time step $o_i+2$ with $n'+m-i'-2$ delays left.
	Hence, the time arc $s'_i \xrightarrow{o_i+1} x_{i}$ is not available, and theh traveler is forced to take the time arcs
\[
	s'_i \xrightarrow{o_i + 2} \bar{x}_i \xrightarrow[o_i+3]{o_i+2} \bar{x}_i^{(1)} \xrightarrow[o_i+4]{o_i+3} \bar{x}_i^{(2)} \ldots \xrightarrow[o_i+n'+m-i']{o_i+n'+m-i'-1} \bar{x}_i^{(n'+m-i'-2)}.
\]
	The vertices $\bar{x}_i^{(k)}$ for all $k \in [n'+m-i'-2]$ have a time arc to the next vertex at two time steps.
	In order to only make the time arc at the latter time step available, the adversary has to delay all $k+1$ time arcs before.
	Thus, to reach $\bar{x}_i^{(n'+m-i'-2)}$ at time step $o_i+n'+m-i'+1$, the adversary would need to use $n'+m-i'-1$ delays, but only $n'+m-i'-2$ are remaining.
	Since the next outgoing time arc at vertex $\bar{x}_i^{(n'+m-i'-2)}$ is at $o_i+n'+m-i'$, the adversary will not use any delays on the vertices above, due to \cref{obs:more_delays_better}.
	The time arc $\bar{x}_i^{(n'+m-i'-2)} \xrightarrow{o_i+n'+m-i'} s_{i+1}$ is the only available time arc from $\bar{x}_i^{(n'+m-i'-2)}$, and thus the traveler will take this time arc.
	Again due to \cref{obs:more_delays_better}, the adversary will not use a delay, since the outgoing time arcs of $s_{i+1}$ are at the earliest at $o_{i+1} + 1 = o_i+n'+m-i'+1$. 
	
	\textit{Case 2:} The adversary does not use a delay on $s_i \xrightarrow{o_i+1} s'_i$ and thus the traveler arrives in $s'_i$ at time step $o_i+1$ with $n'+m-i'-1$ delays left.
	Both outgoing time arcs $s'_i \xrightarrow{o_i + 2} \bar{x}_i$ and $s'_i \xrightarrow{o_i + 1} x_i$ are available.
	However, the traveler cannot choose the first time arc, since the adversary can use all remaining $n'+m-i'-1$ delays to make the traveler arrive in $\bar{x}_i^{(n'+m-i'-2)}$ at $o_i+n'+m-i'+1$, where no time arc for the traveler is available:
\[
	s'_i \xrightarrow{{\color{red} o_i + 3}} \bar{x}_i
	\xrightarrow[{\color{red}o_i+4}]{{\color{gray}o_i+2}} \bar{x}_i^{(1)}
	\xrightarrow[{\color{red}o_i+5}]{{\color{gray}o_i+3}} \bar{x}_i^{(2)} \ldots
	\xrightarrow[{\color{red}o_i+n'+m-i'+1}]{{\color{gray}o_i+n'+m-i'-1}} \bar{x}_i^{(n'+m-i'-2)} 
	\xrightarrow{{\color{gray}o_i+n'+m-i'}} s_{i+1}
\]
	Here, red time labels denote delayed time arcs and gray time labels denote unavailable time arcs for the traveler.
	Hence, the traveler has to traverse the time arcs 
\[
	s'_i \xrightarrow{o_i + 1} x_i \xrightarrow[o_i + 2]{o_i + 1} s_{i+1}.
\]
	Since there is the time arc $s_{i+1} \xrightarrow{o_i + 1} z$ directly to the target vertex, the adversary has to use one delay, in order to delay the arrival in $s_{i+1}$ to $o_i + 2$ and make this time arc unavailable.
	Due to \cref{obs:more_delays_better}, the adversary has no incentive to use two delays to have an arrival time of $o_i + 3$ in $s_{i+1}$, since the outgoing time arcs in $s_{i+1}$ are at earliest at $o_{i+1} + 1 = o_i+n'+m-i' + 1 \ge o_i+3$.
	Hence, the number of remaining delays in $s_{i+1}$ is $n'+m-i'-2$.
	
	In both cases, the traveler cannot take the time arcs to $z'$ at $o_{n+m}+1$ from vertex $x_i$ or vertex $\bar{x}_i$, since the adversary can use a delay, and the only time arc $z' \xrightarrow{o_{n+m}+1} z$ from $z'$ becomes unavailable leading to a loss for the traveler.
	The number of variables quantified by a universal quantifier until the $(i+1)$-st variable is $i'+1$, since the current variable~$x_i$ is universally quantified.
	Hence, in $s_{i+1}$ there are $n'+m-i'-2$ remaining delays, and thus \cref{inv:sel_gad_inv} holds.
\end{proof}

We now present another invariant that states the number of remaining delays at the beginning of the validation gadgets.

\begin{invariant} \label{inv:qbf_game_val_gad_delays}
	For all $i \in [m-1]$, if the start vertex of the $i$-th validation gadget ($v_{i-1}^{(r,m-i)}$ for $i \ge 2$, and $s_{n+1}$ for $i = 1$) is entered, then there are $m-i$ remaining delays, assuming both players use a rational strategy.
\end{invariant}

The correctness of the invariant follows from \cref{obs:qbf_game_val_gad_inv_start} and \cref{lemma:qbf_game_val_gad}.

\begin{observation} \label{obs:qbf_game_val_gad_inv_start}
	\cref{inv:qbf_game_val_gad_delays} holds at $s_{n+1}$, the beginning of the first validation gadget.
\end{observation}
\begin{proof}
	Due to \cref{inv:sel_gad_inv}, there are $n'+m-i'-1$ remaining delays in $s_i$, where $i'$ is the number of variables quantified by a universal quantifier before the $i$-th variable.
	There are $n'$ variables quantified by a universal quantifier.
	Hence, there are $n'+m-n'-1 = m-1$ delays remaining in $s_{n+1}$.
\end{proof}

We now show that the adversary controls the paths taken through the validation gadgets.
The mechanism is similar to the selection gadget for variables quantified by a universal quantifier.

\begin{lemma} \label{lemma:qbf_game_val_gad}
	Assuming both players use a rational strategy, for all $i \in [m-1]$, in the $i$-th validation gadget, the adversary can decide whether the traveler takes the temporal path to vertex $v_i^{(l, m-i)}$ with $0$ delays remaining, or the temporal path to vertex $v_i^{(r, m-i+1)}$ with $m-i-1$ delays left (\cref{inv:qbf_game_val_gad_delays} holds).
\end{lemma}
\begin{proof}
	Due to \cref{inv:qbf_game_val_gad_delays}, the start of the $i$-th validation gadget ($v_{i-1}^{(r,m-i)}$ for $i \ge 2$, and $s_{n+1}$ for $i = 1$) is entered with $m-i$ delays.
	Incoming time arcs are at latest at $o_{n+i}$, and thus the only out-going time arc to $v_i$ at $o_{n+i}+1$ is always available.
	Since it is the only out-going time arc, the traveler has to take this time arc.
	Depending on whether the adversary is delaying this time arc, they can enforce which temporal path is taken by the traveler.
	
	\textit{Case 1:} 
	The adversary uses a delay on $v_{i-1}^{(r,m-i)} \xrightarrow{o_{n+i}+1} v_i$ (or $s_{n+1} \xrightarrow{o_{n+1}+1} v_1$ for $i=1$) and thus the traveler arrives in $v_i$ at time step $o_{n+i}+2$ with $m-i-1$ delays left.
	Hence, the time arc $v_i \xrightarrow{o_{n+i}+1} v_{i}^{(l,1)}$ is not available, and the traveler is enforced to take the time arcs
	\[
	v_i \xrightarrow{o_{n+i} + 2} 
	v_i^{(r,1)} \xrightarrow[o_{n+i}+3]{o_{n+i}+2} 
	v_i^{(r,2)} \xrightarrow[o_{n+i}+4]{o_{n+i}+3} 
	\ldots 
	\xrightarrow[o_{n+i}+m-i+1]{o_{n+i}+m-i} 
	v_i^{(r,m-i)}
	\xrightarrow{o_{n+i}+m-i+1} 
	v_i^{(r,m-i+1)}.
	\]
	The vertices $v_i^{(r,k)}$ for all $k \in [m-i-1]$ have a time arc to the next vertex at two time steps.
	In order to only make the time arc at the latter time step available, the adversary has to delay all $k$ time arcs before.
	Thus, to reach $v_i^{(m-i)}$ at time step $o_{n+i}+m-i+2$, the adversary would need to use $m-i$ delays, but only $m-i-1$ are remaining.
	Since the next outgoing time arc at vertex $v_i^{(r,m-i)}$ is at $o_{n+i}+m-i+1$, the adversary will not use any delays on the vertices above, due to \cref{obs:more_delays_better}.
	The time arc $v_i^{(r,m-i)} \xrightarrow{o_{n+i}+m-i+1} v_i^{(r,m-i+1)}$ is the only available time arc from $v_i^{(r,m-i)}$, and thus the traveler will take this time arc.
	Again due to \cref{obs:more_delays_better}, the adversary will not use a delay, since the outgoing time arcs of $v_i^{(r,m-i+1)}$ are at the earliest at $o_{n+i+1} + 1 = o_{n+i}+m-i+2$. 
	The vertex $v_i^{(r,m-i+1)}$ is the start of the $(i+1)$-st validation gadget, and since the number of remaining delays is $m-i-1$, \cref{inv:qbf_game_val_gad_delays} holds.

	\textit{Case 2:}
	The adversary does not use a delay on $v_{i-1}^{(r,m-i)} \xrightarrow{o_{n+i}+1} v_i$ (or $s_{n+1} \xrightarrow{o_{n+1}+1} v_1$ for $i=1$) and thus the traveler arrives in $v_i$ at time step $o_{n+i}+1$ with $m-i$ delays left.
	Both outgoing time arcs $v_i \xrightarrow{o_{n+i} + 2} v_i^{(r,1)}$ and $v_i \xrightarrow{o_{n+i} + 1} v_i^{(l,1)}$ are available.
	However, the traveler cannot choose the first time arc, since the adversary can use all remaining $m-i$ delays to make the traveler arrive in $v_i^{(m-i)}$ at $o_{n+i}+2$, where no time arc for the traveler is available:
	\[
	v_i \xrightarrow{{\color{red} o_{n+i} + 3}} v_i^{(r,1)}
	\xrightarrow[{\color{red}o_{n+i}+4}]{{\color{gray}o_{n+i}+2}} v_i^{(r,2)}
	\xrightarrow[{\color{red}o_{n+i}+5}]{{\color{gray}o_{n+i}+3}} \ldots
	\xrightarrow[{\color{red}o_{n+i}+m-i+2}]{{\color{gray}o_{n+i}+m-i}} v_i^{(m-i)} 
	\xrightarrow{{\color{gray}o_{n+i}+m-i+1}} v_i^{(r,m-i+1)}
	\]
	Here, red time labels denote delayed time arcs (delay $\delta = 1$ is already added) and gray time labels denote unavailable time arcs for the traveler.
	Hence, the traveler has to take the time arc $v_i \xrightarrow{o_{n+i}+1} v_i^{(l,1)}$.
	Since there is a direct time arc $v_i^{(l,1)} \xrightarrow{o_{n+i}+1} z$ to the target vertex, the adversary has to delay the time arc, so that the arrival in $v_i^{(l,1)}$ is at $o_{n+i}+2$ with a total of $m-i-1$ remaining delays.
	Now, for all $k \in [m-i-1]$ the traveler has to take the time arc $v_i^{(l,k)} \xrightarrow{o_{n+i}+k+1} v_i^{(l,k+1)}$, and the adversary has to delay the time arc, so that the time arc $v_i^{(l,k+1)} \xrightarrow{o_{n+i}+k+1} z$ becomes unavailable.
	Once the traveler arrives in~$v_i^{(l,m-i)}$, there are no delays remaining.
\end{proof}

Now, by repeatedly applying \cref{lemma:qbf_game_val_gad}, we see that the adversary can decide in which of the end-points of the validation gadgets the traveler arrives, and no delays are remaining at the arrival.

\begin{observation} \label{obs:qbf_game_val_gadgets}
	When both players use a rational strategy, being at the start of the first validation gadget $s_{n+1}$ with $m-1$ remaining delays, the adversary can choose exactly one vertex $v_i^{(l,m-i)}$, for $i \in [m-1]$, or $v_{m-1}^{(r,2)}$, that is visited with no delays left when arriving in this vertex.
\end{observation}

Each end-point is corresponding to a specific clause of $\varphi$, the vertex $v_i^{(l,m-i)}$ corresponds to the $i$-th clause, for $i \in [m-1]$, and $v_{m-1}^{(r,2)}$ corresponds to the $m$-th clause.
Now, we can show that Player~1 can reach the target vertex from these vertices if and only if the clause is satisfied under the variable selection done in the selection gadgets.

\begin{lemma} \label{lemma:qbf_game_win_if_phi_true}
	Let $\tilde{x}_1, \tilde{x}_2, \ldots, \tilde{x}_n$ be a truth assignment with $\tilde{x}_i = x_i$ if the vertex $x_i$ has been traversed, and $\tilde{x}_i = \bar{x}_i$ if the vertex $\bar{x}_i$ has been traversed.
	Furthermore, let there be no remaining delays and the current time step at latest at $o_{m+n}$.
	From vertex $v_i^{(l,m-i)}$ the traveler can reach the target vertex $z$ if and only if the $i$-th clause of $\varphi$ is satisfied under the truth assignment $\tilde{x}_1, \tilde{x}_2, \ldots, \tilde{x}_n$.
	Furthermore, from vertex $v_{m-1}^{(r,2)}$ the traveler can reach the target vertex $z$ if and only if the $m$-th clause of $\varphi$ is satisfied under the truth assignment $\tilde{x}_1, \tilde{x}_2, \ldots, \tilde{x}_n$.
\end{lemma}
\begin{proof}
	For $i \in [m]$, let the $i$-th clause be $(l_i^{(1)} \vee l_i^{(2)} \vee l_i^{(3)})$.
	
	($\Leftarrow$): Assume for $i \in [m]$, the $i$-th clause is satisfied.
	Thus, the traveler is in vertex $v_i^{(l,m-i)}$ for $i < m$ or $v_{m-1}^{(r,2)}$ for $i=m$.
	Since the clause is satisfied, there exists a satisfied literal $l_i^{(j)}$ for $j \in [3]$.
	
	(Case 1:) $l_i^{(j)} = x_k$ for some $k \in [n]$.
	The literal being satisfied implies that $\tilde{x}_k = x_k$ and thus, the vertex $x_k$ has been visited.
	Thus, $\bar{x}_k$ has not been visited, since due to \cref{lemma:qbf_game_sel_gad_exist,lemma:qbf_game_sel_gad_univ} exactly one of these vertices are traversed.
	Since $l_i^{(j)} = x_k$, there is a temporal arc $v_i^{(l,m-i)} \xrightarrow{o_{n+m}+1} \bar{x}_k$ (or $v_{m-i}^{(r,2)} \xrightarrow{o_{n+m}+1} \bar{x}_k$ for $i = m$) which is available for the traveler, since $\bar{x}_k$ has not been visited and $o_{n+m}+1$ is larger than the current time step.
	The adversary has no delays left, so the traveler arrives in $\bar{x}_k$ at $o_{n+m}+1$.
	From there the time arcs
	\[
		\bar{x}_k \xrightarrow{o_{n+m}+1} z' \xrightarrow{o_{n+m}+1} z
	\]
	to the target vertex can be taken.
	
	(Case 2:) $l_i^{(j)} = \bar{x}_k$ for some $k \in [n]$.
	The literal being satisfied implies that $\tilde{x}_k = \bar{x}_k$ and thus, the vertex $\bar{x}_k$ has been visited.
	Thus, $x_k$ has not been visited.
	Since $l_i^{(j)} = \bar{x}_k$, there is a temporal arc $v_i^{(l,m-i)} \xrightarrow{o_{n+m}+1} x_k$ (or $v_{m-i}^{(r,2)} \xrightarrow{o_{n+m}+1} x_k$ for $i = m$) which is available for the traveler, since $x_k$ has not been visited and $o_{n+m}+1$ is larger than the current time step.
	The adversary has no delays left, so the traveler arrives in $x_k$ at $o_{n+m}+1$.
	From there the time arcs
	\[
	x_k \xrightarrow{o_{n+m}+1} z' \xrightarrow{o_{n+m}+1} z
	\]
	to the target vertex can be taken.	
	
	($\Rightarrow$): Assume the traveler can get from vertex $v_i^{(l,m-i)}$ (or $v_{m-i}^{(r,2)}$  for $i = m$) to $z$.
	The vertex $v_i^{(l,m-i)}$ (or $v_{m-i}^{(r,2)}$  for $i = m$) has exactly one outgoing time arc for each literal of the $i$-th clause. 
	If $l_i^{(j)} = x_k$ for some $k \in [n]$, then there is a time arc at $o_{n+m}+1$ to vertex $\bar{x}_k$, and if $l_i^{(j)} = \bar{x}_k$ for some $k \in [n]$, then there is a time arc at $o_{n+m}+1$ to vertex $x_k$.
	Since the traveler can reach $z$, one of these time arcs are available and thus the end vertex of this arc $x_k$ or $\bar{x}_k$ has not been visited yet. 
	Thus, there exists a $j \in [3]$, so that if $l_i^{(j)} = x_k$, then $\bar{x}_k$ has not been visited, and if $l_i^{(j)} = \bar{x}_k$, then $x_k$ has not been visited.
	However, if $\bar{x}_k$ has not been visited, the vertex $x_k$ has, and vice versa, since due to \cref{lemma:qbf_game_sel_gad_exist,lemma:qbf_game_sel_gad_univ} exactly one of these vertices are traversed.
	Thus, if $l_i^{(j)} = x_k$, then $x_k$ has been visited and the literal~$l_i^{(j)}$ is satisfied.
	Also, if $l_i^{(j)} = \bar{x}_k$, then $\bar{x}_k$ has been visited, and the literal $l_i^{(j)}$ is satisfied.
	Hence, the $i$-th clause is satisfied.
\end{proof}

Finally, we can prove that our constructed \drpg{}-instance is a yes-instance if and only if $\Phi$ is a yes-instance of \textsc{QBF Game}.
Using the previous lemmas, instead of looking at every single game move, we can reduce the possible actions to two per gadget.

\begin{lemma} \label{lemma:qbf_game_prob_equiv}
	$\Phi = Q_1 x_1. Q_2 x_2.\ldots Q_i x_i. \varphi$, with $Q_i \in \{\exists, \forall\}$ for $i \in [m]$, is a yes-instance of \textsc{QBF Game} if and only if the constructed instance $(\GG = (V,E), s_1, z, \delta = 1, x = n' + m - 1)$	is a yes-instance of \drpg{}.
\end{lemma}
\begin{proof}
	In \drpg{}, assuming both players use a rational strategy, the game can be summarized in the following way:
	In the selection gadgets
	\begin{itemize}
		\item the traveler selects whether vertex $x_i$ or $\bar{x}_i$ is traversed in $\GG$, if $Q_i = \exists$ (\cref{lemma:qbf_game_sel_gad_exist}),
		\item and the adversary selects whether vertex $x_i$ or $\bar{x}_i$ is traversed in $\GG$, if $Q_i = \forall$ (\cref{lemma:qbf_game_sel_gad_univ}),
	\end{itemize}
	in the order of $i=1$ to $n$.
	We use the term \textit{truth assignment in $\GG$} as the truth assignment of the variable $x_i$ to true or false, if the vertex $x_i$ or $\bar{x}_i$ has been traversed in the $i$-th selection gadget, respectively, for all $i \in [n]$.
	Then
	\begin{itemize}
		\item the adversary selects a vertex	corresponding to a clause in $\varphi$ (\cref{obs:qbf_game_val_gadgets}), 
		\item and the traveler wins if and only if the clause is satisfied under the truth assignment in $\GG$ (\cref{lemma:qbf_game_win_if_phi_true}).
	\end{itemize}
	The only freedom of choice for the traveler is thus to select the vertex $x_i$ or $\bar{x}_i$ in selection gadgets for variables quantified by an existential quantifier. 
	The decision is depending on whether $x_j$ or $\bar{x}_j$ has been traversed, for all $j \in [i-1]$.
	Similarly in \textsc{QBF Game}, the only freedom of choice for Player~1 is to select the truth assignment for variables quantified by an existential quantifier. 
	The decision is depending on the truth assignment for the variable $x_j$, for all $j \in [i-1]$.
	Now, assuming a winning strategy for one game, a winning strategy for the other game can be built.
	
	($\Rightarrow$): Assume the traveler has a winning strategy for \drpg{}, telling whether $x_i$ or $\bar{x}_i$ should be traversed, given whether $x_j$ or $\bar{x}_j$ has been traversed, for all $j \in [i-1]$.
	For a \textsc{QBF Game} strategy, given a truth assignment for $x_1, x_2, \ldots, x_{i-1}$, Player~1 assigns $x_i$ to true or false, depending on whether the strategy of \drpg{} tells the traveler to traverse $x_i$ or $\bar{x}_i$, respectively, where for all $j \in [i-1]$, if $x_j$ has been assigned to true, then the vertex~$x_j$ has been traversed, and if $x_j$ has been assigned to false, then the vertex~$\bar{x}_j$ has been traversed.
	Since the traveler's strategy is winning, but the adversary selects a clause in $\varphi$, every clause in $\varphi$ is satisfied under all truth assignments in $\GG$ that resulted by applying the winning strategy.
	The truth assignment for \textsc{QBF Game} is the same as the truth assignment in $\GG$, since the selection followed from the winning strategy of \drpg{}.
	Since all clauses are satisfied, $\varphi$ is satisfied, and Player~1 has a winning strategy for \textsc{QBF Game}.
	
	($\Leftarrow$): Assume Player~1 has a winning strategy for \textsc{QBF Game}, telling whether $x_i$ should be assigned to true or false, given the truth assignment for $x_1, x_2, \ldots, x_{i-1}$.
	For a \drpg{} strategy, given whether the vertex $x_j$ or $\bar{x}_j$ has been traversed, for all $j \in [i-1]$, the traveler traverses $x_i$ or $\bar{x}_i$, depending on whether the strategy of \textsc{QBF Game} tells Player~1 to assign $x_i$ to true or false, respectively, where for $j \in [i-1]$, if the vertex $x_j$ has been traversed, then the truth assignment of~$x_j$ is true, and if the vertex $\bar{x}_j$ has been traversed, then the truth assignment of~$x_j$ is false.
	Since Player~1's strategy is winning, all clauses in $\varphi$ are satisfied under all truth assignments that resulted by applying the winning strategy.
	The truth assignment in $\GG$ is the same as the truth assignment for \textsc{QBF Game}, since the selection followed from the winning strategy of \textsc{QBF Game}.
	Now, whatever clause vertex the adversary chooses, the traveler wins, since the clause is satisfied.
	Hence, the traveler has a winning strategy for \drpg{}.
\end{proof}

Additionally, we show that the construction of $\GG$ can be done in polynomial time.

\begin{lemma} \label{lemma:qbf_game_poly_time_construction}
	Given a \textsc{QBF Game}-instance $\Phi = Q_1 x_1. Q_2 x_2. \ldots Q_n x_n. \varphi$, with $Q_i \in \{\exists, \forall\}$ for all $i \in [n]$, and $\varphi$ having $m$ clauses, the temporal graph $\GG$ can be constructed in $\bigO(n^2 + m^2)$ time.
\end{lemma}
\begin{proof}
	A selection gadget for a variable quantified by an existential quantifier adds four vertices and time arcs to $\GG$.
	A selection gadget for a variable quantified by a universal quantifier adds $\bigO(n + m)$ vertices and time arcs to $\GG$.
	Since there are $n$ selection gadgets this is a total of $\bigO(n^2 + n \cdot m)$ vertices and time arcs.
	Each validation gadget adds $\bigO(m)$ vertices and time arcs to $\GG$.
	Since there are $m$ validation gadgets, this is a total of $\bigO(m^2)$ vertices and time arcs.
	Finally two additional vertices $z'$ and $z$ are added to $\GG$.
	For each clause there are three time arcs to $z'$ ($3\cdot m$ time arcs in total).
	There is a time arc from $x_i$ to $z'$, and from $\bar{x}_i$ to $z'$, for all $i \in [n]$ ($2\cdot n$ time arcs in total).
	There is one additional time arc from $z'$ to~$z$.
	Hence, $\GG$ has $\bigO(n^2 + m^2)$ vertices and time arcs.
	For the construction it is necessary to iterate over $\Phi$ once.
	Hence, the time needed to build~$\GG$ is determined by the size of $\GG$ and takes $\bigO(n^2 + m^2)$  time.
\end{proof}

Due to \cref{lemma:qbf_game_prob_equiv,lemma:qbf_game_poly_time_construction} we have \textsc{QBF Game} $\lepoly$ \drpg{}.
Since \textsc{QBF Game} is \CPSPACE{}-complete, we can conclude \CPSPACE{}-hardness for \drpg{}.

\begin{theorem}\label{thm:pspacehard}
	\drpg{} is \CPSPACE{}-hard.
\end{theorem}

\subsection{PSPACE-containment of \textsc{Delayed-Routing Path Game}}
 \label{sec:drpg_pspace_hard}
Complementing the \CPSPACE{}-hardness of \drpg{} from the previous section, we now show that \drpg{} is containted in \CPSPACE{}, which lets us conclude \CPSPACE{}-completeness of \drpg{}.

We show this by modifying the dynamic program for \drg{} (\cref{sec:drg_poly_alg}) to also save the set of vertices that already has been visited for every state.
This will cause the dynamic programming table to have exponential size; however, we can evaluate it recursively, that is, recomputing every entry when needed.
In this way we only require polynomial space. Formally, we adapt the recursive formula as follows.

Let $I = (G = (V,E), s,z \in V, \delta,x \in \NN)$ be an instance of \drpg{}. 
A game state can be fully described by the $4$-tuple $(v,t,y,V')$, where $v \in V$ is the current vertex, $t \in \NN$ is the current time step, $y\in[0,x]$ is the number of remaining delays, and $V' \subseteq V$ the set of visited vertices.
We may take~$t$~to be from the set $T :=  \{1, \tatime(e) + \tatrav(e), \tatime(e) + \tatrav(e) + \delta \mid e \in E\}$.
The starting game state is $(s,1,x,\emptyset)$.

We define $F:V\times T \times[-1,x]\times 2^{V}\rightarrow \{\true,\false\}$ to indicates for each game state whether the traveler has a winning strategy from that game state.
Denote by $E_t(v) := \{(v, w, t', \lambda) \in E \mid t' \geq t\}$
the set of all time arcs that are available at $v \in V$ at time~$t$ or later. %
Then, for all $v \in V \setminus \{z\}$, $t \in T$, $y \in [0,x]$, and $V'\subseteq V\setminus\{z\}$
the following holds:
\begin{align}
F(z,t,y,V') &= \true, \label{rec:base2}\\
F(v,t,-1,V') &= \true, \label{rec:minus2}
\end{align}
\vspace{-4ex}
\begin{equation}
\begin{split}
F(v,t,y,V') = \bigvee_{e \in E_t(v)} &\big(\tadest(e)\notin V'\land F(\tadest(e), \tatime(e) + \tatrav(e),y,V'\cup\{v\}) \land{} \\ 
 &  F(\tadest(e), \tatime(e) + \tatrav(e) + \delta,y-1,V'\cup\{v\}) \big), \label{rec:general3}
\end{split}
\end{equation}
where the empty disjunction evaluates to \false.

Using \eqref{rec:base2}--\eqref{rec:general3}, we get the following result by evaluating it in a depth-first-search fashion from the starting configuration.

\begin{proposition} \label{cor:drpg_in_pspace}
	\drpg{} is contained in \CPSPACE{}.
\end{proposition}
\begin{proof}
The correctness of our approach can be shown in a way analogous to \cref{lem:DP}.
Note that since the set $V'$ of visited vertices is growing with each move of the traveler, the game cannot run infinitely.
Thus, all entries can be computed by means of \eqref{rec:base2}--\eqref{rec:general3}.

Instead of storing all (exponentially many) entries,
we evaluate $F$ in a depth-first-search fashion from the starting configuration $(s,1,x,\emptyset)$.
This only requires us to keep the current branch of the search tree in memory.
Since each game state requires polynomial space and there are at most $O(|V|)$ moves in a game of \drpg{} (every vertex can be visited at most once),
we only require polynomial space.
\end{proof}

From \cref{thm:pspacehard} and \cref{cor:drpg_in_pspace} we can now conclude that \drpg{} is \CPSPACE{}-complete.
\begin{corollary}
\drpg{} is \CPSPACE{}-complete.
\end{corollary}

\section{Conclusion and Outlook} \label{ch:conclusion}
On the spectrum of delay-related routing problems,
we have studied two extreme (but natural) cases in terms of when information about the delays is made available.
Interestingly, both are polynomial-time solvable, whereas a ``middle ground'' case studied in companion work turned out NP-hard.

It might also seem surprising that \drg{} is efficiently solvable while \drpg{} is \CPSPACE{}-complete.
However, this situation is not unprecedented.
For example, deciding whether a temporal path under waiting time constraints exists (\textsc{$\Delta$-Restless Temporal Path}) is \CNP{}-complete~\cite{TemporalPathsUnderWaitingTime}, while finding temporal walks under waiting time constraints can be done in polynomial time~\cite{BentertHNN20TempWalkWaitingTime}.
Similarly, counting foremost temporal paths is \CSHRPP{}-hard~\cite{rad2017countingPaths}, while counting of foremost temporal walks can be done in polynomial time~\cite{moltCountingEff}.

We remark that instead of delaying edges by increasing their traversal time,
it is also sensible to instead delay their time label.
It can be shown that our results on \drc{} transfer also to this modified version.
For \drg{} the situation is more complicated, we leave this open for future work.

Even more different notions of delays could also be explored.
While in our definitions up to $x \in \NN$ time arcs can be delayed by a fixed integer $\delta$ each,
one could also define an overall ``budget'' $\Delta$ which can be distributed among all time arcs.
Thus, a time arc could be delayed by more than~$\delta$ or more than~$x$~time arcs could be delayed by less than~$\delta$ each.

\clearpage

\bibliography{strings-long,bibfile}

\end{document}